\documentclass[pra,twocolumn,a4paper,nofootinbib]{revtex4-1}

\usepackage{hyperref}
\usepackage{graphicx}
\usepackage{amsmath}
\usepackage{amssymb}
\usepackage{color}
\usepackage{amsthm}
\usepackage{amsfonts}
\usepackage{subfigure}
\usepackage{enumerate}

\newtheorem{theorem}{Theorem}
\newtheorem{lem}{Lemma}

\begin{document}

\title{Fault-tolerant quantum gates with defects in topological stabiliser codes}
\author{Paul Webster}
\affiliation{Centre for Engineered Quantum Systems, School of Physics, The University of Sydney, Sydney, NSW 2006, Australia}
\author{Stephen D.~Bartlett}
\affiliation{Centre for Engineered Quantum Systems, School of Physics, The University of Sydney, Sydney, NSW 2006, Australia}

\date{10 August 2020}

\begin{abstract}
Braiding defects in topological stabiliser codes has been widely studied as a promising approach to fault-tolerant quantum computing. Here, we explore the potential and limitations of such schemes in codes of all spatial dimensions. We prove that a universal gate set for quantum computing cannot be realised by supplementing locality-preserving logical operators with defect braiding, even in more than two dimensions. However, notwithstanding this no-go theorem, we demonstrate that higher dimensional defect-braiding schemes have the potential to play an important role in realising fault-tolerant quantum computing. Specifically, we present an approach to implement the full Clifford group via braiding in any code possessing twist defects on which a fermion can condense. We explore three such examples in higher dimensional codes, specifically: in self-dual surface codes; the three dimensional Levin-Wen fermion mode; and the checkerboard model.  Finally, we show how our no-go theorems can be circumvented to provide a universal scheme in three-dimensional surface codes without magic state distillation. Specifically, our scheme employs adaptive implementation of logical operators conditional on logical measurement outcomes to lift a combination of locality-preserving and braiding logical operators to universality.
\end{abstract}

\maketitle
\section{Introduction}
\label{Sec:Introduction}
Quantum computers can solve certain problems more efficiently than their classical counterparts.  However, the fragility of quantum coherence against noise means that quantum computers will likely require error correction to function on a large scale. Specifically, this necessitates a universal set of logical operators (gates) that can be implemented \emph{fault-tolerantly}, meaning that typical errors remain controlled and correctable throughout the computation. Developing architectures with such a universal, fault-tolerant gate set is challenging, because the most natural approach to fault-tolerant operators, \textit{transversality}, cannot yield a universal set in any quantum error correcting code \cite{Eastin}. 

Topological stabiliser codes  \cite{Bravyi,Pastawski,Webster,BombinNoBraid,Vasmer,BombinGCC,Kubica,Jochym,Campbell,Roberts,Bombin2018,BrownJIT} are a widely-studied and extremely promising class of codes, because they provide protection against general local errors, and their realisation by local Pauli Hamiltonian models provides a route to experimental feasibility. However, even stronger constraints apply to fault-tolerant logical operators in topological stabiliser codes. In particular, locality-preserving logical operators, which are a much more general class than transversal operators, have been proven to be insufficient for universality in any topological stabiliser code \cite{Bravyi,Pastawski,Webster}. A number of ways to construct a universal gate set in topological stabiliser codes have been proposed, including magic state distillation \cite{MSD}, stabiliser state injection \cite{Vasmer}, dimensional jumping \cite{BombinGCC} and just-in-time gauge fixing \cite{Bombin2018,BrownJIT}. However, it remains an open question as to what approach is best in terms of maximising error correction thresholds and minimising overheads. Thus, there is significant value in investigating approaches to universality that are naturally topological.

In this paper, we explore the possibilities and limitations of a natural class of fault-tolerant logical operators in topological stabiliser codes that is more general than locality-preserving logical operators. Specifically, we allow for operators that may be implemented by braiding topological defects. Such operators have been extensively studied in two dimensional codes, with a large range of schemes proposed using holes \cite{Raussendorf1,Raussendorf2,BombinHole,Fowler1,Fowler2,Fowler3,Hastings,Brell,Brown} and twist defects \cite{Brown,Bombin,Bombin2,Kesselring,Scruby,Cong,Yoder,BarkeshliGenon1,BarkeshliGenon2,BarkeshliTheory1,BarkeshliTheory2}. It is known that there exist codes for which this set of logical operators implementable by braiding defects is larger than that of locality-preserving logical operators admitted by the code \cite{Brown,Bombin}. One may thus hope that braiding defects might allow the limitations of locality-preserving logical operators to be overcome. In particular, defects allow abelian models to exhibit non-abelian braiding statistics.  Because there are known examples of non-abelian braiding models that allow for universal fault-tolerant quantum computing \cite{Freedman,Kitaev}, we may wonder if there exist topological stabiliser codes with defects that similarly allow for universality. This possibility seems all the more promising since there are known cases of topological models that do not admit a universal set of logical operators fault-tolerantly, but that do allow for universality when genons (a type of twist defect) are introduced and braided \cite{BarkeshliGenon1}.

Our first result is to show that this approach is fundamentally limited.  Specifically, we prove that the set of operators implementable by braiding defects in a topological stabiliser code cannot be universal.  In addition, this proof extends to the more general case where defect braiding operations are combined with locality-preserving logical operators.  Our no-go result is broad enough to encompass braiding of exotic defects associated with non-Pauli Hamiltonian terms, which locally violate the nature of defect-free stabiliser codes as commuting Pauli models, and which can not generally be braided by using Pauli measurements.  We present the theoretical background underpinning these results in Sec.~\ref{Sec:Theory}. In particular, we review topological stabiliser codes and their properties. We then discuss defects in such codes and how quantum information may be encoded and manipulated using such defects. In Sec.~\ref{Sec:Limitations}, we present and discuss the results constraining the potential of these logical operators. We also show that under appropriate assumptions defect braiding cannot be used to implement non-Clifford logical operators, which provides an even stronger bound within its more restricted domain of applicability.

Having established the limitations in the power of braiding defects, we then consider a number of encodings that saturate and even circumvent these limits. Specifically, in Sec.~\ref{Sec:Schemes} we explore encodings that allow for the full Clifford group to be implemented. We begin this exploration with a review of the well-known example of twists in the two dimensional surface code \cite{Bombin,Brown}. By adapting this scheme, we then present three novel schemes for the Clifford group by braiding in higher dimensional topological stabiliser codes.  These are the first proposed examples of defect-braiding schemes in higher dimensional codes.  They are: a class of even-dimensional surface codes; the three dimensional Levin-Wen fermion model \cite{Levin,Webster} (a three-dimensional topological stabiliser code with point-like fermionic excitations); and the checkerboard model (a code that exhibits fracton excitations) \cite{Vijay,Shirley1,Dua}.  Even in light of our constraints on the power of braiding defects in such encodings, these examples are still worthy of study for a number of reasons. In particular, such higher dimensional schemes can exhibit properties, such as single-shot error correction and self-correction, that cannot be realised in two dimensional schemes. Additionally, supplementing such schemes with magic state distillation provides an approach to achieving a fault-tolerant universal gate set in a more exotic range of codes such as fracton models. They also provide examples of braiding phenomena that are unique to higher-dimensional codes, and may provide inspiration for schemes in more exotic higher dimensional models that are not constrained by our results. Finally, they also provide illustrations of the types of schemes that are constrained by our no-go theorems, beyond the two dimensional schemes studied in previous work.

Finally, we propose a potential mechanism to achieve universality within a topological stabilizer code by adding minimal resources that can circumvent the above no-go theorems.    Specifically, in Sec.~\ref{Sec:Universal Scheme}, we develop a scheme based on encoding qubits in punctures in the three dimensional colour code. By supplementing locality-preserving and braiding logical operators with operators that are conditional on logical measurements, we circumvent the constraining results of Sec.~\ref{Sec:Limitations}. Such conditioning uses non-local classical processing, which is known to allow for previous no-go theorems such as the Bravyi-K{\"o}nig bound to be overcome \cite{BombinSingleShot}. Our scheme is similar in inspiration to that proposed in Ref.~\cite{Vasmer}, with the advantage that braiding defects allows for a universal gate set to be implemented on arbitrarily many logical qubits within a single colour code block without the need for lattice surgery. This scheme illustrates how the limitations on the power of braiding defects in topological stabiliser codes may be overcome.

\section{Topological Stabiliser Codes and Defects}

\label{Sec:Theory}

We begin this section by briefly reviewing the relevant structure and properties of topological stabiliser codes, focussing on the structure of excitations of such codes and locality-preserving logical operators (an important type of fault-tolerant logical operator in such codes).  We then turn to the theory of defects in topological stabiliser codes. Specifically, we define the notion of topological defects, and survey several examples.  We then describe how defects can be used to encode information, the properties of the encoding and finally how braiding defects allows for fault-tolerant logical operators.

\subsection{Topological Stabiliser Codes}
\label{Sec:TSC}

A topological stabiliser code is defined on a lattice of physical qubits in $D \geq 2$ spatial dimensions. To construct a topological stabiliser code, we begin by specifying a stabiliser group of Pauli operators, $\mathcal{S}$. We choose this group so that the following conditions hold. First, it must be abelian and not contain $-I$. Second, it must admit some generating set, $S=\{S_i\} \subseteq \mathcal{S}$, such that all the $S_i$ operators are local with respect to the arrangement of qubits in the lattice. This ensures that the interactions that must be created during encoding are local and that error correction can be performed using only local measurements. Third, any operator that commutes with all elements of $\mathcal{S}$ but is not in $\mathcal{S}$ must be non-local and have support that is topologically non-trivial (that is, it cannot be smoothly deformed to a point). The codespace is then defined to be the +1-eigenspace of $\mathcal{S}$, i.e.,~the set of states stabilised by all elements of $\mathcal{S}$.

For such a topological stabiliser code, we can construct a local Hamiltonian that possesses this codespace as its degenerate ground space. Specifically, a Hamiltonian consisting of commuting Pauli terms may be associated with any stabiliser code by choosing a set of stabiliser generators, $S$, which we can use to define the Hamiltonian $\mathcal{H}=-\sum_{S_i \in S} S_i$. The ground space of this Hamiltonian is the common $+1$-eigenspace of all elements of $S$, which is equivalent to the codespace (i.e.~the space stabilised by $S$). For a topological stabiliser code, the generating set $S$ may be chosen to contain only local operators.  (This is not in general possible for an arbitrary stabiliser code.)

The locality of Hamiltonian terms in this model ensures that local operators, such as those which would describe natural couplings to an environment, can only create local excitations.  In contrast, the degenerate ground states of a topological stabilizer code can only be transformed by operators with support on a manifold (i.e.,~a path or surface) that is topologically non-trivial and cannot be continuously deformed to a point.  Because of this property, a topological code is robust against local noise, and any local error on a topological stabiliser code is correctable. This very general protection against general local noise makes topological stabilizer codes well-suited for use as quantum memories.   

Another useful feature of topological stabilizer codes is that they define a code family on lattices of varying size.  The distance of the code is typically a function of the lattice dimensions.  Choosing a sufficiently large lattice allows for an arbitrarily large code distance to be achieved, while locality of stabiliser generators ensures that the weight of Hamiltonian terms remain constant in this distance. This allows for the effective noise rate on encoded information to be reduced arbitrarily low, provided the physical noise rate is below a particular threshold value \cite{Aharanov,Terhal}.

\subsubsection{Locality-Preserving Logical Operators}
\label{LPLO}
Logical operators act nontrivially on the degenerate ground space of the code, and for a topological code they must have support on a topologically non-trivial manifold.  We now consider the effect of a logical operator acting on a topological stabilizer code, viewed as a quantum logic gate.  To ensure the protection from errors provided by the code is maintained, we restrict our attention to the set of logical operators that can be implemented \textit{fault-tolerantly}. Specifically, this requires that, if implemented on a code where a local error is present, this error must remain local after applying the logical operator. Since local errors in topological stabiliser codes are necessarily correctable, this ensures that correctable errors remain correctable after each action of a logical operator. 

A locality-preserving logical operator is a logical operator that preserves the locality of errors present in the code. More precisely, a logical operator, $\bar{L}$ is locality-preserving if it can be applied by a unitary operator that grows the support of any local operator by at most some constant under conjugation. We note that all transversal gates are necessarily locality-preserving logical operators (since they do not grow the support of operators within a single codeblock). In particular, since logical Pauli operators are transversal in stabiliser codes \cite{Gottesman}, this implies that they are locality-preserving logical operators.

Locality-preserving logical operators may be considered the most direct type of fault-tolerant logical operators that may be implemented in a topological code. They have been extensively studied (see Refs.~\cite{Bravyi,Pastawski,Webster}) and a full classification for a large class of topological stabiliser codes is known \cite{Webster}. We direct the reader to Ref.~\cite{Webster} for further detail on the properties and limitations of locality-preserving logical operators in such codes.

\subsection{Topological Excitations}
\label{TEs}
Excitations of topological stabiliser codes have exotic properties, such as non-trivial braiding statistics, due to the topological ordering of the model. In two spatial dimensions, the localised excitations of a topological stabiliser code are anyons. Each type of anyon may be distinguished by its fusion and braiding rules, and collectively these rules define the anyon model. The surface code provides the simplest example of such an anyon model, admitting four types of anyons: the vacuum ($1$), electric charge ($e$), magnetic flux ($m$) and a composite excitation produced by fusing a charge and a flux ($em$) \cite{Kitaev2}.

When considering higher dimensional codes, it is useful to generalise the notion of an anyon to a broader class of \textit{topological excitations}. Topological excitations are those that admit a well-defined notion of braiding. More precisely, for any topological excitation, there must be at least two spatial dimensions in which the excitation is localised and can freely be propagated by the action of local unitary operators. 

Topological excitations in $D \geq 3$ dimensions are much richer than anyons in two dimensions: they need not be point-like, but may take the form of extended objects; higher-dimensional models can display fracton order; and stable topological excitations that are not energy eigenstates can exist. We briefly describe each of these exotic phenomena.

\subsubsection{Higher Dimensional Topological Excitations} 
First, for a topological stabilizer code in $D\geq 3$ spatial dimensions, excitations need not be point-like to exhibit topological properties analogous to those of anyons. Indeed, excitations may be of any spatial dimension $j$ such that $0\leq j \leq D-2$ and still allow for topological phenomena such as braiding. We refer to such generalised versions of anyons as \textit{eigenstate excitations} \cite{Webster}, to refer to the fact that they are energy eigenstates.  For example, the three dimensional surface code has both electric and magnetic excitations, but with one-dimensional magnetic flux loops along with point-like electric charges. While point-like excitations can only fuse to the vacuum in pairs, higher dimensional excitations can individually contract to the vacuum provided they form closed loops or surfaces (or more generally, hypersurfaces). 

Topological excitations in topological stabiliser codes can be understood as boundaries of locality-preserving logical operators. Specifically, define a \textit{restricted locality-preserving logical operator} to be the restriction of a locality-preserving logical operator to a compact region, such that it has a boundary in the bulk of the code. Applying a restricted locality-preserving logical operator gives rise to an excitation localised to the boundary \cite{Yoshida,Webster}. If this excitation is free to move in at least two spatial dimensions, then it will be a topological excitation.

This relationship is very familiar for eigenstate excitations. For example, in the two dimensional surface code, a logical Pauli operator is implemented by a string of Pauli operators between an appropriate pair of boundaries or defects, or across a topologically non-trivial loop. Restricting a string of this kind so that it has endpoints in the bulk of the code gives rise to a pair of anyons at these endpoints. More generally, a restricted logical Pauli operator has an excitation at its boundary. This excitation will be an eigenstate excitation, since the Pauli operators that give rise to it necessarily commute or anticommute with each stabiliser meaning that the state that results from their application will be an energy eigenstate. Since logical Pauli operators are necessarily locality-preserving (as discussed in Sec.~\ref{LPLO}), this is an instance of the general phenomenon described above.

\subsubsection{Fracton Order} 
Higher dimensional models may also exhibit fracton order, with a variety of possible associated phenomena~\cite{Chamon,HaahCubic,Vijay,Dua}. In particular, such models may admit excitations known as \textit{fractons} that are point-like, but cannot move freely throughout the code without giving rise to additional point-like excitations \cite{Dua}. Due to their lack of mobility, analysing topological properties of such excitations is a challenge. However, models with fracton order can also admit point-like excitations that are free to move within a plane, but unable to move otherwise without producing additional excitations. Such excitations are referred to as \textit{planons}  \cite{Dua,Shirley1,Shirley2}. Planons, and higher dimensional generalisations, are examples of topological excitations that have no analogue in two dimensional codes, and add to the rich structure of excitations in higher dimensional codes.

\subsubsection{Non-Eigenstate Excitations} 
Another feature of topological stabiliser codes in $D\geq 3$ spatial dimensions is the possibility for stable topological excitations that are not energy eigenstates, but rather \textit{non-eigenstate excitations} of at least one dimension \cite{Else,BombinS,Yoshida,YoshidaB,YoshidaC,Webster,Kubica2018}. The existence of such excitations has been inferred independently from multiple different considerations. In particular, we note three complementary perspectives on this class of excitations. Mathematically, the existence of \textit{Cheshire charge} -- a delocalised charge across a loop-like excitation that cannot be decomposed into local charges -- is implied by analysis of three dimensional topological stabiliser codes as corresponding to braided fusion 2-categories \cite{Else}. Excitations carrying Cheshire charge are non-eigenstate excitations. Independently, study of the action of non-Clifford locality-preserving logical operators on Pauli errors in topological stabiliser codes has led to the identification of \textit{linking charges} \cite{BombinS}. These linking charges are also equivalent to non-eigenstate excitations (and, indeed, carry Cheshire charge). The route to understanding non-eigenstate excitations that is most relevant to our work, however, is due to Yoshida \cite{Yoshida,YoshidaB,YoshidaC}. From this perspective, non-eigenstate excitations are best understood as boundaries of non-Pauli locality-preserving logical operators.

As an example, consider the three-dimensional colour code, which admits a two-dimensional locality-preserving logical Clifford phase operator, $\bar{S}$ \cite{BombinNoBraid,Yoshida}. A restricted $\bar{S}$ locality-preserving logical operator has a loop-like one-dimensional boundary, and an excitation is created on this boundary. This excitation is not an energy eigenstate, since the physical Clifford operators that give rise to this excitation are superpositions of Pauli operators that neither commute nor anticommute with the $X$-type stabilisers at the boundary of the region. This excitation is a superposition of vacuum and point-like $e$ excitations along the one-dimensional boundary.  Note that this excitation is stable, in that it cannot be projected onto an eigenstate excitation by local operators -- such projection requires operators to act on its entire interior (a type of SPT protection \cite{Yoshida}). For this reason, it must be considered a distinct excitation that cannot be realised by simply assembling a configuration of eigenstate excitations, just as its corresponding locality-preserving logical operator is a superposition of logical Pauli operators but cannot be realised by a product of Pauli operators. As with eigenstate excitations, this non-eigenstate excitation can freely move through the code and can be characterised by its braiding statistics \cite{Yoshida}. It must thus be considered a distinct type of topological excitaiton.

More general non-eigenstate excitations can be understood similarly. Specifically, non-eigenstate excitations arise at the boundary of non-Pauli restricted locality-preserving logical operators. Since non-Pauli locality-preserving logical operators are necessarily of dimension $k\geq 2$ \cite{Bravyi,Webster}, the boundaries of such restricted locality-preserving logical operators cannot be point-like, and so can only support topological excitations in codes of more than two dimensions. As noted, general non-eigenstate excitations arise at boundaries of non-Pauli restricted locality-preserving logical operators to a region. This excitation cannot be an energy eigenstate, since non-Pauli operators necessarily have some Pauli operator with which they neither commute nor anticommute. However, as with the above example, such operators cannot be locally decomposed into eigenstate excitations. In addition, non-eigenstate excitations can be characterised by their braiding statistics with eigenstate excitations, that correspond to their (nested) commutation relations between their corresponding logical operators and logical Pauli operators \cite{Yoshida}.

\subsection{Examples of Defects in Topological Stabiliser Codes}
\label{DefectsinTSC}
By way of introduction, in this section we survey several examples of topological defects in topological stabiliser codes. The aim of this review is to provide the reader with an appreciation for the range of known topological defect schemes to which our work applies. We defer a formal definition of defects and defect encodings to Sec.~\ref{Sec:Encoding}

We emphasise that our survey below does not exhaust the full range of defects, encodings and braiding schemes that are possible in topological stabiliser codes. We further emphasise that our results presented in Sec.~\ref{Sec:Limitations} do not depend on the specific properties of any of the example topological defects surveyed below, and indeed apply to all topological defects that satisfy the definition in Sec.~\ref{Sec:Encoding}.

\subsubsection{Holes}
A \textit{hole} (or \textit{puncture}) is a compact region of the code where the stabilisers are removed from the Hamiltonian (so that this region is equivalent to the vacuum). The resulting code boundary may condense topological excitations, in which case the hole is a topological defect. For example, the surface code admits two types of boundaries; \textit{rough} boundaries which can condense electric charge ($e$) excitations and \textit{smooth} boundaries which can condense magnetic flux ($m$) excitations \cite{BravyiKitaev}. This code correspondingly admits two types of holes (those with rough and smooth boundaries) that condense $e$ and $m$ type excitations, respectively. The boundaries of holes are closed loops (or, more generally, hypersurfaces) of some finite size, independent of the size of the code. This compactness allows holes to be moved around through the bulk of the code, allowing for braiding. 

To illustrate how such holes can be used to encode quantum information, consider the example of the two dimensional surface code. As discussed above, holes in a two dimensional surface code can be constructed to have rough boundaries that allow $e$ excitations, but not $m$ excitations, to condense. Consider introducing a pair of holes of this type into a surface code. Assume that there is no total net charge across the pair. We may then associate a logical qubit with the charge in the region of either of the holes; i.e.~the parity of $e$ excitations in this region. Provided the holes are separated by a distance, $d$, this parity can only be changed by an operator of weight at least $d$.  (Here and in what follows, all distances are measured in terms of number of qubits along a path in the lattice.) Similarly, provided each hole has a circumference (i.e.~shortest path enclosing the hole) of at least $d$, the parity can only be distinguished by an operator of weight at least $d$. We can thus define a logical qubit with computational basis states $|\bar{0}\rangle$ and $|\bar{1}\rangle$ corresponding to even and odd parity of $e$ excitations respectively in the region of either hole.

The Pauli operators for this logical qubit are then naturally transversal. Specifically, a logical $\bar{X}$ operator may be implemented by transferring an $e$ excitation between the pair of holes, using a string of $Z$ operators between the holes. A logical $\bar{Z}$ operator may be implemented by braiding an $m$ excitation around either hole, since this results in a phase of $-1$ if and only if the parity of $e$ in this region is odd. This operation corresponds to a string of $X$ operators around a hole. These logical operators are illustrated in Fig.~\ref{Holes}.

\begin{figure}
\centering
\includegraphics{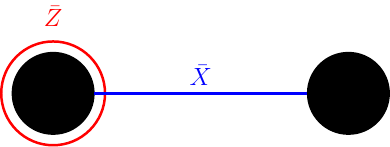}
\caption{An illustration of encoding a logical qubit in a pair of `rough boundary' holes in a 2D surface code. Blue lines represent the path of an excitation, $e$, that can condense at the boundary of a hole. Red lines represent paths of excitation $m$. The Pauli logical operators for this encoded qubit are illustrated. \label{Holes}}
\end{figure}

Schemes for realising a universal set of logical operators on logical qubits encoded in holes have been developed for both the two dimensional surface \cite{Fowler2} and colour codes \cite{Fowler3}. These schemes use braiding to achieve an entangling CNOT operator between logical qubits. However, these braiding operations must then be supplemented in two ways. First, to realise single qubit logical operators, braiding must be supplemented by locality-preserving logical operators. These can be used to achieve a Hadamard operator in the surface code \cite{Fowler2} and the full single qubit Clifford group in the colour code \cite{Fowler3}. Second, locality-preserving and braiding logical operators must be supplemented by some additional type of operator to realise universality. In both the surface and colour code this may be done by using magic state distillation to implement a non-Clifford gate such as the $\bar{T}$ gate \cite{Fowler2,Fowler3}.

\subsubsection{Domain Walls and Twists}
\label{Sec:Twists}

A second type of defect arises from considering boundaries between two copies of the same code. Such boundaries are referred to as domain walls. Unlike code boundaries, they can be transparent, meaning they do not absorb topological excitations, but instead may transform them from one type of topological excitation into another. They may be introduced by applying a restricted locality-preserving logical operator to a region of the code.  Domain walls in a $D$-dimensional code may be $k$-dimensional for any $k$  such that $1\leq k<D$ \cite{Yoshida,Webster}. 

As a simple example of a domain wall, consider the two dimensional colour code. A set of local stabilizer generators of this code are given by products of $X$ operators and products of $Z$ operators around each plaquette of the lattice \cite{ColourCode}.  The two dimensional colour code possesses a symmetry given by the action of the Hadamard operator, which interchanges $X$ and $Z$ operators, on all qubits. If this symmetry is applied to some compact region of the code, then the result is a gapped boundary around this region -- a domain wall. This boundary is a defect, since some but not all of the qubits in the support of stabilisers that lie along it are acted on by the symmetry. This means that these stabilisers will have a different structure after the symmetry has been applied from stabilisers in the rest of the code, and so the symmetry breaks the translational invariance of the code. Excitations that cross this domain wall are transformed in accordance with the symmetry applied to the interior region. For example, in the two dimensional colour code with Hadamard symmetry, $e$ and $m$ excitations correspond to $Z$ and $X$ type errors respectively. These excitations are thus interchanged upon crossing the wall into the region where $Z$ and $X$ have been interchanged.

As described so far, however, domain walls do not allow for the condensation of excitations, and so do not qualify as topological defects by our definition. Such condensation is made possible by terminating domain walls at boundaries of their own.  Such domain wall boundaries form a type of topological defect, called twists. For details on how to construct such twists, we refer readers to the procedure presented in Ref.~\cite{BarkeshliTheory2}.

To see that twists allow topological excitations to condense, and so are topological defects, consider an example using the Hadamard symmetry in the two dimensional colour code.  A pair of $e$ excitations may be created from the vacuum, and one of these excitations brought across the domain wall, converting it to an $m$. This $e$ and $m$ may then be brought together by allowing them to meet beyond the twist, producing the composite excitation $em$.  Thus, this twist in the two dimensional colour code can condense $em$ excitations.  More generally, domain walls may be classified by their action on excitations of the code. To see this, denote by $b^*$ an excitation that annihilates to the vacuum with $b$. (If $b$ is a quasiparticle, $b^*$ is the antiparticle of $b$. More generally, we refer to $b^*$ as the anti-excitation of $b$.) A domain wall that takes $a \to b$ will then allow composite excitation $a^*b$ to condense. This process is illustrated in Fig.~\ref{TwistCondensation}.

\begin{figure}
\centering
\includegraphics{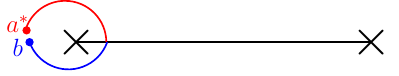}
\caption{The condensation of excitations at domain walls with twists. The black line is a domain wall that transforms excitation $a$ to excitation $b$. Here, a (red) particle-antiparticle pair, $a,a^*$, condenses in the vacuum. The excitation $a$ crosses the wall to become (blue) excitation, $b$, while $a^*$ does not. Both then come together to give the excitation $a^*b$. Note that without the twists, such condensation would not be possible since excitations could then only come together if the number of times they crossed the domain wall was of the same parity. \label{TwistCondensation}}
\end{figure}

Schemes for quantum computing with twists have been studied for a range of two dimensional codes. These codes include the surface code \cite{Bombin,Brown}, the colour code \cite{Kesselring,Scruby}, subsystem colour codes \cite{Bombin2} and the $\mathbb{Z}_3$ abelian quantum double model (which generalises the surface code to qutrits) \cite{Cong}. All of these schemes allow only Clifford gates by braiding, and require additional methods, such as magic state distillation \cite{MSD} or topological charge measurement \cite{Cong}, to realise universality. We note that genons, a particular type of twist defect, have also attracted study for their braiding properties across a range of models, and allow for universality by braiding in codes that are more exotic than topological stabiliser codes~\cite{BarkeshliGenon1,BarkeshliGenon2}. For more on the general theory of defects in two dimensional topological stabiliser codes, we refer the reader to Refs.~\cite{BarkeshliTheory1,BarkeshliTheory2}.

\subsection{Formalism for Encoding in Topological Defects}
\label{Sec:Encoding}
We now present a formal definition of topological defects and a general formalism for encoding quantum information in topological defects. 

To define defects, assume that we initially have a topological stabiliser code that is translationally invariant. A \textit{defect} is defined to be a $k$-dimensional region introduced to this code where this translational invariance is broken, with $0\leq k<D$. This can be viewed as a region of the code where the Hamiltonian terms are altered. Unlike the original stabilisers, these altered Hamiltonian terms are not required to be Pauli operators but they must be local. We refer to a defect at which topological excitations can condense as a \emph{topological defect}, as this condensation allows the defect to carry topological charge. We assume that only topological excitations condense at topological defects, to ensure that the defect scheme is fundamentally topological.

To specify how quantum information is encoded in topological defects, we consider a topological stabiliser code supported on a topologically trivial manifold without boundaries. Such a code has no ground state degeneracy (and so encodes no logical information) without defects.\footnote{This assumption is only made to allow us to focus attention on logical information encoded in defects without the complication of additional logical information. Once the scheme for encoding of logical information into these defects has been described, we may transfer this scheme into the bulk of a code on a more general manifold without affecting our conclusions.} Consider introducing a topological defect that allows for the condensation of some abelian group of topological excitations $G$. We can now study sets of excited states that differ from the vacuum only by the presence of topological excitations, $g \in G$, and which can be locally connected through condensation at the defect.  Specifically, we can identify a subspace of dimension $|G|$ of such states, including the vacuum. 

For clarity, from here on we assume that we are working in the special case where $G=\mathbb{Z}_2^k=\langle g_1, ..., g_k\rangle$ for some $k \geq 1$ so that we deal only with binary labelings of these states, but our examples and results naturally generalise to other abelian groups.  Our subspace can then be viewed as $k$ qubits, and (uniquely up to a change of basis) we may define the computational basis states of qubit $m$ distinguished by the parity of excitation $g_m$ at the defect.

However, such an encoding does not provide protection against errors, since all such states can then be interchanged by local errors in the region around the defect. To encode topologically protected quantum information in a ground state degeneracy, we require that logical states are all degenerate and can only be interchanged by operators that are not local to the defect. This can be expressed as the constraint that we require that, for logical states, the net charge of the whole code remains neutral (i.e.~the outcome of fusing all excitations present in the code, including those at topological defects, is the vacuum). Under this constraint, excitations that are created at the defect must be annihilated elsewhere in the code to maintain a logical state. This process by which topological excitation $g_m$ is created at the defect and then is annihilated elsewhere corresponds to the logical operator $\bar{X}_m$. 

In a code with no boundaries on a topologically trivial manifold, there are two possibilities for where this annihilation may take place -- in the vacuum, or at other defects. These two possibilities give rise to two types of encoding of logical information in defects, which we now consider in turn for the purpose of illustration.  We emphasise, however, that more general defect encoding schemes to which our results apply may involve arbitrarily many topological defects, and include logical qubits specified using combinations both types of encodings.

\subsubsection{Two Defect Encoding}
\label{Sec:TDE}
We consider first the case of a topological excitation $g_m$ that is created at one topological defect and annihilated at another. This gives rise to an encoding specified by two topological defects. The $|\bar{0}\rangle_m$ and $|\bar{1}\rangle_m$ states may be associated with even and odd parity of $g_m$ respectively on either of the two topological defects. This encoding is used for all the two-dimensional schemes reviewed in the previous subsection.

We note that, in dimension greater than two, $g_m$ may not be point-like.  Such higher dimensional topological excitations arise at the boundary of a restricted locality-preserving logical operator of more than one dimension.  Since such a boundary can be a single connected object, these excitations can thus be created individually from the vacuum. Thus, to ensure that the parity of $g_m$ remains the same between the pair of defects, we must thread an additional defect through the pair used for encoding that prevents the excitation being created from the vacuum. An example of this in three dimensions is shown in Fig.~\ref{Fig:DoubleDefect}.

\begin{figure}
\centering
\includegraphics{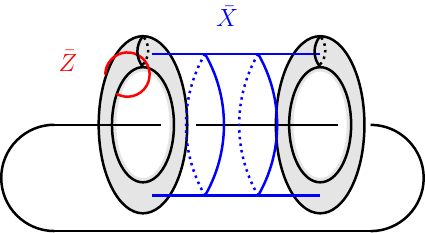}
\caption{An example of a logical qubit encoded in a pair of topological defects. The two tori are holes in the three dimensional surface code at which magnetic fluxes may condense. The additional loop threaded through is a puncture (which may be viewed like a narrow tube) at which magnetic fluxes cannot condense. The computational basis states of the qubit correspond to the parity of fluxes at either of the toroidal holes, which is well-defined since the threaded defect prevents these fluxes from being created from the vacuum. The $\bar{X}$ operator is a surface corresponding to propagating a flux between the holes. The logical $\bar{Z}$ is realised by braiding an electric charge around either hole. \label{Fig:DoubleDefect}}
\end{figure}

\subsubsection{Single Defect Encoding}

In more than two dimensions, it is also possible to encode a logical qubit in a single topological defect, provided the topological excitation used to define the logical qubit can be absorbed into the vacuum. Specifically, consider a defect allowing for the condensation of a topological excitation that is not point-like, for example a loop-like excitation. Such an excitation can be created at a single defect and then shrunk until it is annihilated in the vacuum. Thus, we may associate the computational basis states of a logical qubit with the presence or absence of this excitation on the single defect. An example of this type of encoding is shown in Fig.~\ref{Fig:SingleDefect}.

\begin{figure}
\centering
\includegraphics{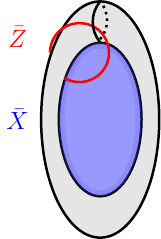}
\caption{An example of a logical qubit encoded in a single defect. The defect is a toroidal hole in a 3D toric code that can condense one dimensional $m$ excitations. The logical $\bar{X}$ operator is implemented by a process by which an $m$ is created at a defect and then shrinks to be annihilated in the vacuum. This is equivalent to applying a membrane of $X$ operators that meets the torus in a topologically non-trivial loop. The logical $\bar{Z}$ is implemented by braiding an $e$ excitation around the defect. This is equivalent to a loop of $Z$ operators around the torus. \label{Fig:SingleDefect}}
\end{figure}

\subsection{Fault-Tolerant Logical Operators for Encodings in Defects}
Logical qubits encoded in topological defects admit multiple classes of fault-tolerant logical operators. We here describe these classes and the relationships between them.

\subsubsection{Topological Locality-Preserving Logical Operators}
\label{TLPLOs}
{Topological locality-preserving logical operators (TLPLOs) are logical operators that can be implemented by the propagation of topological excitations. Specifically, a TLPLO is a logical operator that admits an implementation by which a topological excitation (or set of topological excitations) is created, is dragged along a topologically non-trivial path through the code, and then is annihilated to return to the code space. This creation and annihilation may take place at topological defects (e.g.~the $\bar{X}$ operators in Figs.~\ref{Holes} and \ref{Fig:DoubleDefect}). Alternatively, it may take place in the vacuum, with either the creation and annihilation of a pair of point-like excitations (e.g.~$\bar{Z}$ in Figs.~\ref{Holes}, \ref{Fig:DoubleDefect} and \ref{Fig:SingleDefect}) or of a single topological non-trivial loop or higher-dimensional excitation. We also allow for creation from the vacuum and annihilation at a defect or vice versa (e.g.~the $\bar{X}$ operator in Fig.~\ref{Fig:SingleDefect}). In each case, the corresponding TLPLO is a logical operator, since it preserves the ground space. 

TLPLOs are necessarily locality-preserving logical operators. This is because, by contrast to the propagation of more general topological objects, such as non-abelian anyons and defects, propagation of topological excitations in a topological stabiliser code is necessarily locality-preserving. This remains true even in the presence of defects. Indeed, for propagation of a topological excitation to be well-defined, we require that it must preserve the Hamiltonian on regions through which the propagation occurs. This implies that propagation must preserve the locality of Hamiltonian terms. In the bulk of the code (i.e.~away from topological defects), all Hamiltonian terms are local and no excitations can condense and so all local errors must have some set of local Hamiltonian terms (which may be non-Clifford) with which they do not commute. Hence, local errors must also remain local under the propagation of topological excitations. In the region around topological defects, we allow topological excitations to condense and so there can be operators local to this region that commute with all Hamiltonian terms. However, the only processes that require topological excitations to come close to topological defects are creation and annihilation, which are necessarily locality-preserving. Thus, TLPLOs are locality-preserving. Since the creation, propagation and annihilation of topological excitations is implemented by unitary operators, this implies that TLPLOs are indeed locality-preserving logical operators.

We note that the set of TLPLOs acting on any set of logical qubits encoded in topological defects always contains a non-trivial element. Specifically, the logical $\bar{X}$ acting on each logical qubit is a TLPLO, as described in Sec.~\ref{Sec:Encoding}.

\subsubsection{More General Locality-Preserving Logical Operators}
\label{LPLOs2}
Locality-preserving logical operators (as introduced in Sec.\ref{LPLO}) are a more general class of fault-tolerant logical operators than TLPLOs. In particular, we note that while all TLPLOs are locality-preserving logical operators, not all locality-preserving logical operators can be implemented by propagating topological excitations. For example, locality-preserving logical operators that have support on the whole of a code are not TLPLOs. Relevant examples of such operators include the single-qubit logical Clifford gates required to achieve the full Clifford group in schemes in the two-dimensional surface code \cite{Fowler2} and colour code \cite{Fowler3}, and the three-dimensional $\overline{\text{CCZ}}$ operator in the three-dimensional defect setup described in Sec.~\ref{Sec:Universal Scheme}.

We can study locality-preserving logical operators by considering their action on TLPLOs. In particular, locality-preserving logical operators map TLPLOs to TLPLOs. To see this, note that locality-preserving logical operators necessarily map a localised excitation, $a$, to another localised excitation, $b$ \cite{Webster}. Moreover, if $a$ is a topological excitation, then it must be propagated by local unitary operators (as discussed in Sec.~\ref{TEs}). Such local unitary operators are mapped to local unitary operators under the action of a locality-preserving logical operator. Therefore, $b$ is necessarily propagated by local unitary operators. Also, a path followed by $a$ through the code is mapped to a topologically equivalent path followed by $b$ under the action of a locality-preserving logical operator. Thus, if $a$ is free to move freely along paths in a space of at least two spatial dimensions, then the same must be true of $b$. Hence, if $a$ is a topological excitation then $b$ must also be a topological excitation. Thus, a TLPLO implementable by the creation, propagation and annihilation of a topological excitation $a$ is mapped to a TLPLO implementable by the corresponding creation, propagation and annihilation of a topological excitation $b$.

\subsubsection{Logical Operators by Braiding Defects}
\label{BLOs}
We define a notion of braiding topological defects analogous to that for braiding topological excitations. Specifically, braiding topological defects consists of any process by which the positions of defects in the code are altered smoothly. So that this may be considered a topological process, we assume that defects remain sufficiently large and well-separated that the distance of encoded qubits is preserved throughout. To ensure fault tolerance, we also require that correctable errors remain correctable after braiding has been performed. This can be ensured by moving defects sufficiently gradually that error correction cycles are separated by a time in which the defect moves a length constant in the encoding distance.

The most well-studied braiding process is \textit{code deformation}, wherein the Hamiltonian is adiabatically transformed in such a way as to gradually change the positions of defects \cite{BombinMD}. However, our analysis and results are independent of which particular method is used for braiding. As an aside, we note a recent proposal for a low time-overhead approach to braiding non-abelian anyons~\cite{Zhu,Zhu2}. If this approach could be adapted to the braiding of defects considered here, our analysis would still apply to operators implemented in this way. By contrast, note that processes that involve discontinuous deformations, such as lattice surgery \cite{Horsman}, are not considered to be braiding.

A braiding process that preserves the ground space of the code implements a logical operator, which we refer to as a \textit{braiding logical operator}. Since the braiding is done sufficiently slowly, the code is in the ground space of the Hamiltonian corresponding to the current defect configuration throughout, and so such logical operators correspond to processes for which the final Hamiltonian is the same as the initial Hamiltonian. This is equivalent to saying that the final defect configuration is indistinguishable from the original configuration. Examples of processes that can implement braiding logical operators are illustrated in Fig.~\ref{Fig:braid}. 

We note that, as evident in Fig.~\ref{fig:braid3}, such processes in models of more than two dimensions may be more complicated and varied than in two dimensional codes. Indeed, in general such processes can involve defects of any spatial dimension, and can involve more than two defects at a time \cite{Wang}.

\begin{figure}
\centering
\subfigure[One point-like defect braided around another point-like defect. \label{fig:braid1}]{\label{fig:a}\includegraphics[width=0.15\textwidth]{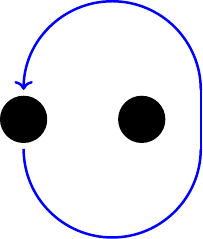}}
\subfigure[One point-like defect is exchanged with another (identical) point-like defect. \label{fig:braid2}]{\label{fig:b}\includegraphics[width=0.15\textwidth]{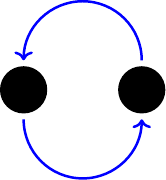}}
\subfigure[One loop-like defect is braided around another loop-like defect. \label{fig:braid3}]{\label{fig:b}\includegraphics[width=0.15\textwidth]{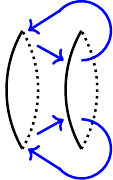}}
\caption{Examples of processes that can implement braiding logical operators. \label{Fig:braid}}
\end{figure}

We can study braiding logical operators by considering their action on TLPLOs. Specifically, as we now show, braiding logical operators map TLPLOs to TLPLOs. To see this, for a defect $\mathcal{D}$ we distinguish two cases -- where the support of the TLPLO intersects with $\mathcal{D}$, and where it does not:
\begin{enumerate}[(i)]
\item \textit{TLPLO does not intersect with $\mathcal{D}$:} If the TLPLO does not intersect with $\mathcal{D}$, then as $\mathcal{D}$ is braided the TLPLO will simply be deformed so that it continues not to intersect with $\mathcal{D}$. The effect of this is not necessarily trivial. For example, if $\mathcal{D}$ is moved to a new position, then the support of the TLPLO will be correspondingly transformed to enclose the defect at this new position (e.g.~as shown in Fig.~\ref{7a}). However, it is clear that it will remain a TLPLO.
\item \textit{TLPLO intersects with $\mathcal{D}$:} If the TLPLO intersects with $\mathcal{D}$, then the effect of braiding $\mathcal{D}$ is that the path followed by the topological excitation to implement the TLPLO is deformed so that it also follows the path along which $\mathcal{D}$ is braided (e.g.~as shown in Fig.~\ref{7b}). Thus, the TLPLO is mapped to a TLPLO implemented by propagating the topological excitation along this new path.
\end{enumerate}

We note that the argument for case (ii) assumes that the topological excitation implementing a TLPLO remains a topological excitation even when additional braiding around topological defects is appended. This assumption  is not trivial, since topological excitations that are braided around topological defects (such as twists) can be mapped to distinct excitations. However, we note that if a topological excitation $a$ were transformed to an object $b$ that were not a topological excitation under braiding around a topological defect, then this would allow for the condensation of the non-topological excitation $a^*b$ at the defect (where $a^*$ denotes the anti-excitation of $a$), analogously to the example of twists in Sec.~\ref{Sec:Twists}. This is inconsistent with the definition of topological defects presented in Sec.~\ref{Sec:Encoding}. Thus, a topological excitation that is braided around a topological defect remains a topological excitation, and so the transformed logical operator corresponds to a TLPLO.

We thus conclude that TLPLOs are indeed mapped to TLPLOs under the action of braiding logical operators, as for locality-preserving logical operators. In the next section we show that this has consequences for the range of braiding logical operators implementable in a code.

\begin{figure}
\centering
\subfigure[Exchange of left and right of defects maps a TLPLO (red) enclosing the left defect to one enclosing the right defect. \label{7a}]{
\includegraphics{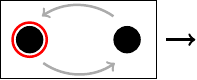}
\includegraphics{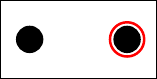}
}
\subfigure[Action of braiding of the top-right defect around the bottom right on TLPLO (blue). The TLPLO initially corresponds to transferring a topological excitation between the top defects. After the braid, it corresponds to this transferral along with the excitation being braided around the bottom-right defect. \label{7b}]{
\includegraphics{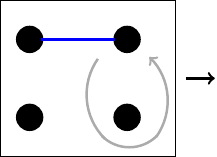}
\includegraphics{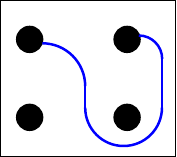}
}
\caption{Examples of the action of braiding logical operators on TLPLOs.}
\end{figure}

\section{Limitations of Braiding Defects}

\label{Sec:Limitations}

Braiding topological defects in topological stabilizer codes can be used to perform fault-tolerant logic gates.  However, the logic gates that can be performed through braiding are highly constrained, as we prove in this section. We begin by proving a very general result:  that the set of logical operators given by any combination of braiding and locality-preserving logical operators cannot be universal.  That is, the limitations on performing universal fault-tolerance given by Eastin-Knill~\cite{Eastin} and Bravyi-K{\"o}nig~\cite{Bravyi} cannot be overcome, even when adding the ability to braid topological defects.  We briefly consider how these results may be naturally generalised to the broader class of abelian quantum double models. Finally, we note how additional natural assumptions regarding the encoding can yield an even stronger constraint, that all braiding logical operators are contained in the Clifford group.

\subsection{Universality No-Go Result}
\label{Sec:NoGo1}

Here we prove our main result, that a universal gate set cannot be achieved by using only locality-preserving logical operators and braiding logical operators in a topological stabiliser code of any spatial dimension. This proof builds on the result of Bravyi and K{\"o}nig \cite{Bravyi} that the group of locality-preserving logical operators admitted by a topological stabiliser code cannot be universal. Specifically, we show that the addition of logical operators implementable by braiding topological defects is insufficient to allow for universality.

We first prove two lemmas, and then the theorem itself.

\begin{lem}
The group of topological locality preserving logical operators (TLPLOs) acting on any finite set of logical qubits encoded in defects is finite.\label{Lem2}
\end{lem}
\begin{proof}
We show first that the set of topological excitations in a topological stabiliser code is finite. Consider a topological stabiliser code with periodic boundary conditions. On such a code, each topological excitation has a corresponding TLPLO implemented by propagating the excitation around the code. Thus, the set of topological excitations cannot be larger than the set of TLPLOs. However, the set of TLPLOs is contained in the set of locality-preserving logical operators, which must be finite \cite{Bravyi}. Thus, the set of topological excitations is finite.

The set of topological defects used to encode a finite set of logical qubits is necessarily finite. For a finite set of topological defects, the set of topologically non-trivial processes a topological excitation can undergo is finite. The set of topological excitations is independent of the presence of topological defects, and so must remain finite even in the presence of defects.  Hence, the set of TLPLOs is finite.
\end{proof}

\begin{lem}
Let $\bar{U}$ be a logical operator implementable by a product of braiding and locality-preserving logical operators, and $\bar{A}$ be topological locality-preserving logical operator (TLPLO). Then, $\bar{U}\bar{A}\bar{U}^\dag$ is a TLPLO. \label{Lem1}
\end{lem}
\begin{proof}
As discussed in Sec.~\ref{BLOs}, braiding logical operators map TLPLOs to TLPLOs, and so if $\bar{A}$ is a TLPLO and $\bar{B}$ is a braiding logical operator then $\bar{B}\bar{A}\bar{B}^\dag$ is a TLPLO. Also, as discussed in Sec.~\ref{LPLOs2}, locality-preserving logical operators map TLPLOs to TLPLOs, and so if $\bar{A}$ is a TLPLO and $\bar{L}$ is a locality-preserving logical operator, then $\bar{L}\bar{A}\bar{L}^\dag$ is a TLPLO. If $\bar{U}$ is a product of braiding and locality-preserving logical operators, it acts on TLPLO $\bar{A}$ by a sequence of conjugations by braiding or locality-preserving logical operators. Thus, $\bar{U}\bar{A}\bar{U}^\dag$ must be a TLPLO.
\end{proof}

\begin{theorem}
The set of logical operators implementable by any product of locality-preserving logical operators and braiding logical operators in a topological stabiliser code cannot be universal. \label{Theorem1}
\end{theorem}
\begin{proof}
By Lemma \ref{Lem2}, the set of TLPLOs acting on any finite set of logical qubits encoded in defects is finite and (as discussed in Sec.~\ref{TLPLOs}) this set always contains a non-trivial element since $\bar{X}_i$ acting on each logical qubit is a TLPLO. By Lemma \ref{Lem1}, products of braiding and locality-preserving logical operators permute this finite set containing a non-trivial element. Thus, the set of products of locality-preserving and braiding logical operators cannot be dense in the set of logical operators.  This means that it is not possible to construct such a set of operators that can be used to approximate any logical operator arbitrarily precisely. So, by definition, the set of combinations of locality-preserving and braiding logical operators cannot be universal.
\end{proof}

Theorem \ref{Theorem1} shows that including logical operators implemented by braiding topological defects is not sufficient to lift the set of locality-preserving logical operators to universality in a topological defect setup in any topological stabiliser code. The sets of logical operators considered in proving this result are represented in Fig.~\ref{Sets}.

\begin{figure}
\includegraphics{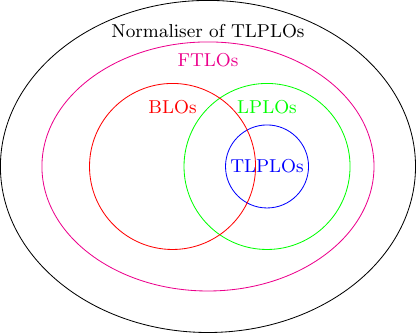}
\caption{A Venn diagram representing relevant groups of logical operators. The group of fault-tolerant logical operators (FTLOs) we consider is defined to be the group of operators implementable by products of locality-preserving logical operators (LPLOs) and braiding logical operators (BLOs). Lemma \ref{Lem1} shows that this group is contained in the normaliser of the group of TLPLOs. Lemma \ref{Lem2} shows that the set of TLPLOs is finite. Combined with the fact that the group of TLPLOs is non-trivial, this implies that the normaliser of the group of TLPLOs is not universal, and hence that the set of FTLOs cannot be universal -- the result of Theorem \ref{Theorem1}. Similarly, for the standard encoding introduced in Sec.~\ref{Sec:NoGo2}, Lemma \ref{Lem3} shows that the group of TLPLOs is the logical Pauli group, which implies that the set of FTLOs is contained in the normaliser of the Pauli group (the Clifford group) -- the result of Theorem \ref{Theorem3}.} \label{Sets}
\end{figure}

The above proof of Theorem \ref{Theorem1} is expressed in terms of the action of braiding and locality-preserving logical operators on TLPLOs, since this is the most direct way to show the result. For clarity, we note that this can also be understood in terms of the action of braiding and locality-preserving logical operators on logical states of encoded qubits. To illustrate this with the example of a single logical qubit, consider the state $|\bar{+}\rangle$ stabilised by TLPLO $\bar{X}$. Lemma \ref{Lem1} implies that this state can only be mapped to another state stabilised by a TLPLO. Since the set of TLPLOs is finite, this implies that there is only a finite set of states to which $|\bar{+}\rangle$ can be mapped and so arbitrary rotations of the Bloch sphere cannot be approximated arbitrarily well by products of locality-preserving and braiding logical operators. This translation of the proof into the language of states can similarly be performed for the general case of multiple logical qubits.

\subsection{Generalisation to Abelian Quantum Double Models}
In this paper we primarily consider topological stabiliser codes. However, we may also consider whether Theorem 1 may be generalised to include other classes of topological quantum error correcting codes. For example, Escobar-Vel{\' a}squez \textit{et al.}~have studied braiding holes in two dimensional Dijkgraaf-Witten theories, which include topological stabiliser codes as well as more general abelian and non-abelian quantum double models \cite{Escobar}. They showed that the set of gates implementable by braiding holes in such codes cannot be universal. In light of our results, it is natural to ask whether this result could be generalised to more general defects and to models of higher spatial dimensions.

We here argue that our results can at least be generalised to all topological defects implemented in codes in the class of abelian quantum double models of all spatial dimensions. We confine our attention to this class since the topological excitation structure and locality-preserving logical operators they admit have been studied and classified \cite{Webster} and shown to be entirely analogous to that of topological stabiliser codes. This means that the theory of Sec.~\ref{Sec:Theory} and results used to prove Theorem \ref{Theorem1} hold entirely and straightforwardly analogously for abelian quantum double models. The class of codes is also interesting, since the power of braiding its defects has been explored in previous work. In particular, the technique of topological charge measurement has been shown to provide universality in the double of $\mathbb{Z}_3$ when used to supplement braiding logical operators \cite{Cong}.

Abelian quantum double models are defined by specifying some abelian group, $A$ \cite{Kitaev2}. Generalised logical Pauli operators are defined such that $\bar{X}$ and $\bar{Z}$ operators each generate groups isomorphic to $A$. The generalised Clifford group for a given $A$ can be defined to be the normaliser of the generalised Pauli group. With these definitions, we may then state the following result. The proof of this result is entirely analogous to that presented for Theorem \ref{Theorem1}.

\begin{theorem}
The set of logical operators implementable by any combination of locality-preserving logical operators and braiding logical operators in an abelian quantum double model cannot be universal.
\end{theorem}

\subsection{No-Go Result for Standard Encoding}
\label{Sec:NoGo2}
Theorem \ref{Theorem1} is a very general result.  It applies to all of the varied encodings considered in Sec.~\ref{Sec:Encoding}, including some exotic encodings that are very different from those for which quantum computing with defects has been studied in detail.  We now consider whether a stronger result can be found for the more standard scenario for quantum computing with defects. Here, we construct an encoding, which we call the \emph{standard encoding}, that naturally generalises the well-studied two-dimensional defect encoding schemes while still allowing for a range of exotic phenomena that emerge only in higher-dimensional codes.  For the standard encoding, we prove an even more restrictive no-go theorem than Theorem \ref{Theorem1}. Specifically, we prove that logical operators implementable by locality-preserving logical operators and braiding logical operators on qubits encoded using the standard encoding are necessarily contained in the Clifford group.

The standard encoding is defined by imposing the following two assumptions. First, all logical qubits are encoded using a \emph{two defect encoding}, as described in Sec.~\ref{Sec:TDE}.  We impose this assumption because such an encoding is more naturally suited to braiding than a single defect encoding, as it allows for the two defects used in the encoding to be braided around each other to implement braiding logical operators on the encoded qubit. Second, we assume that, for each topological excitation $g_m$ that is used to define a logical qubit (as described in Sec.~~\ref{Sec:TDE}), there exists a topological excitation $b_m$ that gives a phase of $-1$ when braided with $g_m$. This assumption is imposed so that each logical qubit has a logical $\bar{Z}$ operator that is a TLPLO, which is implemented by braiding the topological excitation $b_m$ with either of the encoding defects.  This ensures that the property that logical Pauli operators are TLPLOs, which is shared by all standard approaches to encoding information in topological stabiliser codes (with \cite{Raussendorf1,Raussendorf2,BombinHole,Fowler1,Fowler2,Fowler3,Hastings,Brown,Bombin,Bombin2,Kesselring,Scruby} or without \cite{Webster,Kitaev2,BravyiKitaev,ColourCode,BombinNoBraid,Bombin4,Terhal} defects) in all spatial dimensions, is preserved in the standard encoding into defects. For examples of the standard encoding, we refer the reader to Sec.~\ref{Sec:Schemes}, in which all encodings satisfy these assumptions and which illustrate the range of phenomena possible within the scope of the standard encoding in $D>2$ dimensions.

With the standard encoding, Lemma~\ref{Lem2} can be significantly strengthened. Specifically, the finite group of TLPLOs implied by Lemma~\ref{Lem2} is the logical Pauli group, as shown by the following lemma.
\begin{lem}
All TLPLOs acting on a set of logical qubits encoded using the standard encoding are logical Pauli operators. \label{Lem3}
\end{lem}
\begin{proof}
Let $\bar{U}$ be a TLPLO in a $D$-dimensional code. We show first that $\bar{U}$ must be contained in the Clifford hierarchy (indeed in the $D$th level of the hierarchy) and then proceed to show it must in fact be in the Pauli group by induction on the Clifford hierarchy.

If $\bar{U}$ is outside the $D$th level of the Clifford hierarchy, then there exists a set of $D$ Pauli TLPLOs, $\{\bar{P}_1,\ldots \bar{P}_{D}\}$, whose sequential group commutator with $\bar{U}$ is non-trivial \cite{Pastawski,YoshidaC}. This implies that the corresponding defect-free code with periodic boundary conditions has a corresponding set of $D+1$ TLPLOs implemented by propagating the same topological excitations as for $\bar{U}$ and $\{\bar{P}_1,\ldots \bar{P}_{D}\}$ around a topologically non-trivial path whose sequential group commutator is also non-trivial. This implies that the defect-free code has a locality-preserving logical operator outside the $D$th level of the Clifford hierarchy \cite{Pastawski,YoshidaC}, which contradicts Ref.~\cite{Bravyi}.

We now prove the result by induction on the Clifford hierarchy. Specifically, assume all TLPLOs in the $k$th level of the Clifford hierarchy are logical Pauli operators. Then, if $\bar{U}$ is in the $(k+1)$th level of the Clifford hierarchy, its commutator with all single qubit logical Pauli operators must be logical Pauli operators, and so $\bar{U}$ must be in the Clifford group. We thus need only to prove the case where $\bar{U}$ is a Clifford TLPLO.

In this case, for an arbitrary logical qubit, $j$, encoded in a pair of topological defects, $\mathcal{D}$ and $\mathcal{D}'$, there exist logical Pauli operators $\bar{P}= \alpha\prod_i \bar{X}_i^{a_i}\bar{Z}_i^{b_i}$ and $\bar{Q}=\beta\prod_i \bar{X}_i^{c_i}\bar{Z}_i^{d_i}$ (for $\alpha,\beta \in \mathbb{C}$ and $a_i,b_i,c_i,d_i \in \mathbb{Z}_2$) such that $P=\bar{X}_j\bar{U}\bar{X}_j\bar{U}^\dag$ and $Q=\bar{Z}_j\bar{U}\bar{Z}_j\bar{U}$. Since the support of the group commutator of two operators must be contained in their intersection, this implies that $\text{supp}(\bar{P}) \subseteq \text{supp}(\bar{X}_j) \cap \text{supp}(\bar{U})$ and $\text{supp}(\bar{Q}) \subseteq \text{supp}(\bar{Z}_j) \cap \text{supp}(\bar{U})$

However, by construction, $\text{supp}(\bar{X}_j)$ does not enclose any defect in a topologically non-trivial manifold and $\text{supp}(\bar{Z}_j)$ does not connect any defect pairs, so $b_i=c_i=0$ for all $i$. Also, $\text{supp}(\bar{X}_j)$ does not connect any pair of defects except for $\mathcal{D}$ and $\mathcal{D}'$ and $\text{supp}(\bar{Z})$ also does not enclose any pair of defects except for $\mathcal{D}$ or $\mathcal{D}'$, so $a_i=d_i=0$ for all $i\neq j$. Additionally, since $\bar{U}$ is a TLPLO, it corresponds to the path of a topological excitation that necessarily can be freely moved in some direction perpendicular to that in which $\mathcal{D}$ and $\mathcal{D}'$ are separated and to some direction in which $\mathcal{D}$ is localised. Thus, for any particular implementations of $\bar{X}_j$ and $\bar{Z}_j$ we can deform $\bar{U}$ such that $\text{supp}(\bar{P}) \subseteq\text{supp}(\bar{X}_j) \cap \text{supp}(\bar{U})$ does not connect defects $\mathcal{D}$ and $\mathcal{D}'$ and $ \text{supp}(\bar{Q}) \subseteq\text{supp}(\bar{Z}_j) \cap \text{supp}(\bar{U})$ does not enclose $\mathcal{D}$ or $\mathcal{D}'$. Hence, $a_j=d_j=0$.

Thus, $\bar{P}=\alpha\bar{I}$ and $\bar{Q}=\beta\bar{I}$ and so the commutators of $\bar{U}$ with all single qubit logical operators is trivial up to a phase. Thus, $\bar{U}$ is a logical Pauli operator.
\end{proof}

Applying the approach used to prove Theorem \ref{Theorem1} (and summarised in Fig.~\ref{Sets}) to the standard encoding with this stronger lemma, we find that all combinations of locality-preserving and braiding logical operators are in the Clifford group.
\begin{theorem}
The set of logical operators implementable by any product of locality-preserving and braiding logical operator on logical qubits encoded in defects using the standard encoding is contained in the Clifford group. \label{Theorem3}
\end{theorem}
\begin{proof}
By Lemma \ref{Lem1}, any product of braiding and locality-preserving logical operators, $\bar{A}$, maps TLPLOs to TLPLOs under conjugation. By Lemma \ref{Lem3}, the group of TLPLOs is contained in the logical Pauli group and, as all logical Pauli operators are TLPLOs, this implies that the group of TLPLOs is equal to the logical Pauli group. Hence, $\bar{A}$ maps logical Pauli operators to logical Pauli operators under conjugation, and so is contained in the Clifford group.
\end{proof}

Theorem \ref{Theorem3} shows that the set of operators implementable by products of braiding and locality-preserving logical operators on qubits using the standard encoding is highly restricted. In particular, schemes based on braiding defects in two dimensional topological stabiliser codes to implement the Clifford group (such as with twists in the surface code \cite{Bombin,Brown}) have long been known, and Theorem \ref{Theorem3} shows that such schemes cannot be improved upon with the standard encoding.

While this conclusion is consistent with previous results in two dimensional topological stabiliser codes \cite{Bravyi}, it is surprisingly restrictive in higher dimensions. Specifically, we note that for any spatial dimension $D\geq 2$, there exist $D$-dimensional topological stabiliser codes (without defects) that admit locality-preserving logical operators that are strictly in the $D$th level of the Clifford hierarchy \cite{Kubica,Webster}. One might expect that an analogous result would hold for braiding defects -- that a corresponding defect scheme could be constructed in a $D$-dimensional code that allowed a braiding logical operator in the $D$th level of the Clifford hierarchy. This expectation is furthered by the greater range of defects and braiding phenomena possible in higher dimensional codes~\cite{Mesaros,Else,Wang,Yoshida}. However, our result shows that, with the standard encoding, all braiding logical operators are confined to the second level of the Clifford hierarchy, and hence that this is not possible.

The restrictiveness of Theorem \ref{Theorem3} can be understood to reflect a tradeoff between locality-preserving logical operators and braiding logical operators. Specifically, because braiding logical operators permute the group of TLPLOs (as shown in Lemma \ref{Lem2}), the structure of this group constrains the power of braiding in a model. In particular, with the standard encoding the group of TLPLOs is the logical Pauli group (as shown in Lemma \ref{Lem3}), and it is this structure that restricts braiding logical operators to the Clifford group. Interestingly, this tradeoff is reminiscent of that identified for gates implemented by braiding anyons in two-dimensional topological quantum field theories in Ref.~\cite{Beverland}. Indeed, that work observed that models that admit the Pauli group as locality-preserving logical operators have operators implemented by braiding anyons contained in the Clifford group, analogous to our result. By contrast, the more general range of encodings accounted for by Theorem \ref{Theorem1} allows for logical $\bar{Z}$ operators that are not TLPLOs and thus allows for braiding logical operators that are less constrained. However, since even these highly general schemes still must have logical $\bar{X}$ operators as TLPLOs, braiding logical operators still cannot be universal, as shown by Theorem \ref{Theorem1}. Our results thus show how the essential structure of topological stabiliser codes constrains braiding logical operators.

\section{Schemes for Implementing the Clifford Group by Braiding}
\label{Sec:Schemes}
In this section, we explore examples of schemes to implement the full Clifford group fault-tolerantly in topological stabilizer codes.  In particular, we generalise the scheme of Ref.~\cite{Brown} for implementing the Clifford group by braiding twists in the two-dimensional surface code to higher dimensional codes. In doing so, we demonstrate a very general result: that any code with twists that can condense a \textit{generalised fermion} can be used to implement the full Clifford group by braiding.  By a generalised fermion, we mean an eigenstate excitation of any spatial dimension with the property that exchanging a pair of such excitations gives a phase of $-1$. As the Clifford group can be used to achieve universality when supplemented by magic state distillation \cite{MSD}, this generalisation gives a framework for universal quantum computation in a large class of topological stabilizer codes.

In this section, we first review braiding twists in the two dimensional surface code, following Ref.~\cite{Brown}, as the motivating example for this result.  We then provide and describe three examples of twist setups using domain walls in codes of more than two dimensions that have twists that condense generalised fermions. Specifically, we first generalise the example of the two dimensional surface code to \textit{self-dual surface codes}. These are surface codes for which the logical Pauli $\bar{X}$ and $\bar{Z}$ operators are of the same dimension or, equivalently, where electric and magnetic excitations are of the same dimension. This includes the two-dimensional surface code, as well as instances of all even spatial dimensions. While these generalised schemes can only be realised in more than three dimensions, they are worthy of consideration because they offer the potential for self-correction and display braiding phenomena that do not appear in braiding in two dimensional codes.  As our second example, we explore braiding in the three dimensional Levin-Wen fermion model \cite{Levin,Haah} (a three-dimensional topological stabiliser code with point-like fermionic excitations). This is an interesting example for several reasons, including that it offers an example of a three dimensional code that allows for the Clifford group to be implemented fault-tolerantly. Finally, we investigate the checkerboard model. This stabilizer code is an example of a three dimensional foliated fracton model, and so illustrates how our results extend to codes beyond those which admit descriptions as topological quantum field theories. This scheme is based on using planons to encode logical qubits in one dimensional twists. It offers promise as a potential example which could be used to guide future work on performing fault-tolerant gates in fracton models.

In each of the three examples, we show explicitly how the single qubit Clifford group may be implemented. These results can be extended to the Clifford group on any number of encoded qubits by introducing ancilla holes, analogous to the construction in Ref.~\cite{Brown}; however we omit these details for the sake of brevity. We note that the schemes we consider all use the standard encoding and so we have already seen in Sec.~\ref{Sec:NoGo2} that they cannot allow for non-Clifford logical operators.

\subsection{Review:  Two Dimensional Surface Code}
\label{2DSCTwist}

The two-dimensional surface code admits a domain wall that interchanges $e$ and $m$ type excitations, corresponding to the locality-preserving logical Hadamard operator \cite{Bombin,Webster}. This wall allows a composite $em$ type excitation to condense at a twist at one of its endpoints. A pair of such walls with twists at their ends can allow for a single logical qubit to be encoded. Specifically, we encode the states $|\bar{0}\rangle$ and $|\bar{1}\rangle$ as corresponding to an even and odd parity of $em$ excitations in one of the domain walls. The logical Pauli $\bar{X}$ operator is then realised by a process that allows an $em$ excitation to be transferred between the two walls. The logical Pauli $\bar{Z}$ operator is realised by a process that results in a phase of $-1$ if the parity of $em$ excitations is odd. This can correspond to enclosing either of the domain walls in a loop traced out by an $e$ or $m$ excitation. An example of each of these Pauli logical operators is shown in Fig.~\ref{DWEncoding}. We note that this encoding in twists is different from conventional surface code encodings in boundaries or holes, which do not allow for the full Clifford group by locality-preserving logical operators or braiding.

\begin{figure}
\centering
\includegraphics{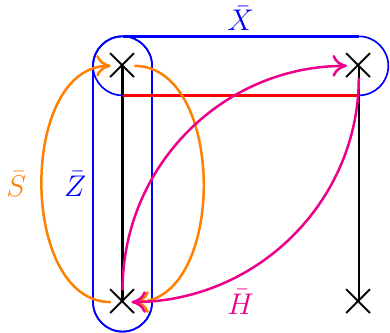}
\caption{The logical operators acting on the logical qubit encoded in two dimensional surface code twists. The Pauli logical operators are locality-preserving logical operators corresponding to paths of excitations of the code. Blue lines represent the path of an $m$ excitation and red lines represent the path of an $e$ excitation. The $\bar{Z}$ operator distinguishes the parity of $em$ excitations in the left domain wall and $\bar{X}$ transfers an $em$ excitation between the domain walls. The Clifford logical operators (orange and magenta) are implemented by braiding, specifically by exchanging twists.\label{DWEncoding}}
\end{figure}

We may then implement logical Clifford operators by interchanging twists to interchange logical Pauli operators \cite{Brown}. Specifically, the $\bar{X}$ and $\bar{Z}$ logical operators may be interchanged by a braiding process with the effect of swapping a pair of twists diagonally (for example the top-right and bottom-left in Fig.~\ref{DWEncoding}). This braid implements the logical Hadamard operator, $\bar{H}$. Similarly, swapping a pair of twists vertically (for example the two left twists) interchanges the logical $\bar{X}$ and $\bar{Y}$ operators (up to a phase of $-1$) and so realises the logical Clifford phase gate, $\bar{S}$. The $\bar{S}$ gate can also be understood by noting that it introduces a phase of $\pm i$ if and only if the parity $em$ on the domain wall is odd, since the excitation is a fermion being rotated by $\pi$. These logical operators are shown in Fig.~\ref{DWEncoding}. Together these logical operators generate the full single qubit Clifford group.

\subsection{Self-Dual Surface Codes}
We now present our first new example of a setup that allows for logical operators by braiding in more than two dimensions. It is a generalisation of the two dimensional surface code scheme we have just presented. Specifically, we show how a surface code with equal dimensions of its logical $\bar{X}$ and $\bar{Z}$ operators, and suitably defined twist-like topological defects, can admit the full Clifford group by braiding. 

We begin by focussing on the example of the four dimensional surface code (the simplest case beyond two dimensions). The four dimensional surface code has one dimensional loop-like $e$ and $m$ type excitations, and admits a three dimensional domain wall that interchanges these excitations \cite{Webster,Roberts4DSC}.  To generalise the setup used for the two dimensional surface code, we construct two pairs of concentric (two dimensional) tori as higher-dimensional twists, with the three dimensional volume between each pair as domain walls. A three dimensional cross-section of this is sketched in Fig.~\ref{4DSC}. Such twists allow for loop-like $em$ excitations to condense, which are generalised fermions. We note that to ensure that these $em$ excitations cannot condense from the vacuum away from these twists, we require an additional defect to be threaded through the twist setup, as shown in Fig.~\ref{4DSC}. This is not required in the two-dimensional case, since its point-like $em$ excitations cannot condense individually from the vacuum. However, such threaded defects are generally required when higher-dimensional excitations are used to encode, since such excitations can condense individually from the vacuum as closed loops (or more generally, hypersurfaces).

A logical qubit may be encoded in the parity of $em$ loops in one of these domain walls. The logical $\bar{X}$ operator is implemented by a process where a pair of $m$ loops is condensed from the vacuum, one of these loops passes through the pair of domain walls to change to $e$ and then back to $m$, then grows to be larger than the outer twists before returning back to annihilate with the other condensed $m$. Note that this process does indeed have the effect of passing an $em$ between the walls. It also creates a torus that encloses the two outer twists, analogously to how the two dimensional setup had a logical $\bar{X}$ operator that encloses the top two twists, as shown in Fig.~\ref{DWEncoding}. The logical $\bar{Z}$ is also a torus. It is implemented by a process by which a pair of $m$ loops condense passing through the hole in the left domain wall but not the right one. One of these $m$ loops then rotates through the dimension perpendicular to the three dimensional subspace in which $\bar{X}$ is embedded. This $\bar{Z}$ operator then encloses the two left twists and so is analogous to the $\bar{Z}$ operator in the two dimensional setup which also enclosed a pair of twists from the same domain wall. A cross-section with these Pauli logical operators are illustrated in Fig.~\ref{4DSC}.

\begin{figure}
\centering
\includegraphics{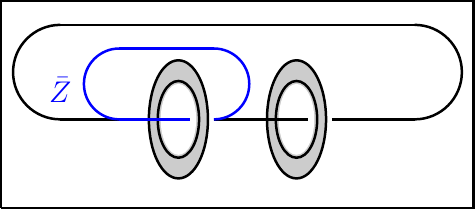}
\includegraphics{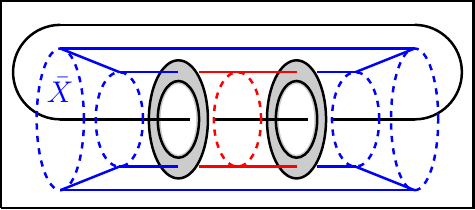}
\includegraphics{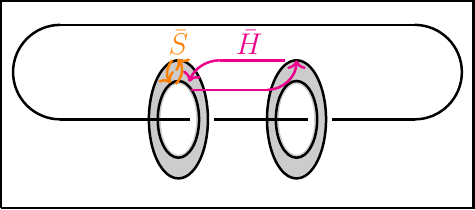}
\caption{A three dimensional cross-section of the twist setup in the four dimensional surface code. The fourth dimension can be considered to be parameterised by an angle, in which all pictured objects except for the $\bar{X}$ operator are extended. Twists are tori, shown here as black circles. Domain walls are the volume between the pair of tori (shown here in grey).
 The logical $\bar{Z}$ and $\bar{X}$ operators are both tori, as shown in the top two pictures. Blue lines indicate paths of loop-like $m$ excitations and red lines indicate paths of loop-like $e$ excitations.The bottom picture shows the Clifford logical operators that can be implemented by braiding, with magenta and orange arrows indicating the movement of defects to implement $\bar{H}$ and $\bar{S}$, respectively. \label{4DSC}} 
\end{figure}

Clifford logical operators by braiding can now be understood by identifying twists of this setup with those of the two dimensional setup. Specifically, inner twists here correspond to bottom twists in the two dimensional setup and left twists in this setup correspond to left twists in the two dimensional setup. Pictured in this way, it is straightforward to see that swapping the left twists will implement the logical phase operator, $\bar{S}$, since it interchanges the two twists enclosed by $\bar{Z}$. Interchanging the inner-left with outer-right twist will implement the logical Hadamard operator, $\bar{H}$, since it swaps the twists that are enclosed by $\bar{X}$ but not $\bar{Z}$ and vice versa. Thus, the setup does indeed allow for the full single qubit Clifford group to be implemented by braiding.

This picture can be generalised to a $2k$-dimensional surface code with $(k-1)$-dimensional $e$ and $m$ excitations. Specifically, we choose two concentric pairs of hyper-surfaces of the form $S^{k-1}\times S^{k-1}$ as twists, where $S^k$ is a $k$-dimensional sphere and the regions between each concentric pair to be domain walls. The $\bar{X}$ and $\bar{Z}$ operators will then be surfaces of the form $S^{k-1}$ enclosing appropriate twists. The $\bar{S}$ and $\bar{H}$ operators can be implemented by appropriate interchanges of the twists.

While the scheme laid out here cannot be realised in less than four spatial dimensions, it is of some theoretical interest. One avenue of interest is that all surface codes considered here are self-correcting. We may expect that property to carry over to their braiding schemes, since all logical operators are of at least two spatial dimensions. Properties of this braiding scheme may thus offer guidance on how self-correction may be achieved in more general braiding schemes. Moreover, this example may provide more guidance on how to generalise two dimensional braiding schemes to higher dimensions.

\subsection{Three Dimensional Levin-Wen Fermion Model}
The three dimensional Levin-Wen fermion model is a three dimensional topological stabiliser code that admits point-like fermionic excitations \cite{Levin, Haah}, which we label $e$. This means that it is not equivalent to the three dimensional surface code, although the structure of its excitations is similar; both codes admit a point-like $e$ and a loop-like $m$. Defined on a cube with boundaries (similarly to the planar surface code) it encodes two or three logical qubits, depending on whether the length of the cube is an even or odd number of qubits. This code with boundaries admits only (two-dimensional) $\overline{\text{CZ}}$ operators between each pair of logical qubits as locality-preserving logical operators \cite{Webster}. The $\overline{\text{CZ}}$ logical operator corresponds to a (one-dimensional) domain wall that allows for the fermionic $e$ excitations to condense. Thus, we may expect that it should admit the full Clifford group by braiding. We show that this is the case. We note, however, that unlike three-dimensional surface codes, this model does not admit a locality-preserving logical $\overline{\text{CCZ}}$ operator \cite{Webster}, and so this scheme does not allow for a universal gate set.

We construct a setup to allow this braiding analogous to the two dimensional surface code. Specifically, we have two domain walls, each ending in a pair of point-like twists. The computational basis states for a logical qubit can then be encoded in the parity of $e$ excitations in either one of these domain walls. The logical $\bar{Z}$ operator is then a surface enclosing one of the domain walls traced out by an $m$ loop growing from the vacuum passing around the domain wall and then shrinking back to the vacuum. The logical $\bar{X}$ operator is implemented by an $m$ loop growing from the vacuum, then passing so that one point passes through a domain wall so that an $e$ is appended, then passing through the other wall so that the $e$ may annihilate with another $e$ that is attached, and finally shrinking back to the vacuum. These Pauli logical operators are illustrated in Fig.~\ref{LWFM}. We may then immediately see that the single qubit Clifford group may be implemented analogously to in the two dimensional surface code. Specifically, interchanging the two left twists will implement the phase operator, $\bar{S}$, since it introduces a phase of $\pm i$ if and only if the parity of $e$ in the left domain wall is odd. Interchanging the bottom-left and top-right twists implements the Hadamard operator, $\bar{H}$, since it swaps the logical $\bar{X}$ and $\bar{Z}$ operators.

\begin{figure}
\centering
\includegraphics{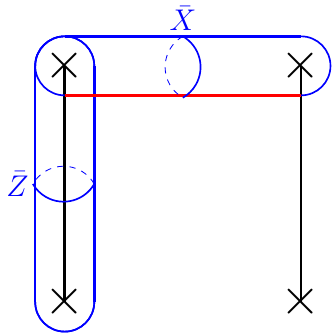}
\caption{The logical operators for a logical qubit encoded in the twists of the three dimensional Levin-Wen fermion model. The twists are points at the end of one dimensional domain walls in a three dimensional space. The logical $\bar{Z}$ operator is a surface of $X$ operators across a sphere enclosing the left twists. The logical $\bar{X}$ operator is a surface of $X$ operators across a sphere enclosing the top twists with a line of $Z$ operators appended between the domain walls. \label{LWFM}}
\end{figure}

This scheme is interesting for a number of reasons. Firstly, it is an instance of how the Clifford group may be implemented without state injection or dimensional jumping in a three dimensional topological stabiliser code. Specifically, the Levin-Wen fermion model allows for the Clifford group to be implemented by braiding point-like defects.  This is interesting as it is known that no translationally-invariant (i.e.~without defects) and scale-symmetric three dimensional topological stabiliser code can admit the full Clifford group by locality-preserving logical operators \cite{Webster} and so offers an example of where braiding defects can allow for fault-tolerant operators not admitted otherwise. The scheme is also notable as an instance of a code where truly string-like Pauli logical operators are relatively rare (since they can only be realised by strings between the pair of twists). Nonetheless, we may expect that the code is not self-correcting, since membrane operators implementing logical Paulis must only be of length $d$ in one dimension and can be arbitrarily small in the other dimension. The interesting properties of this three dimensional braiding setup may inspire and guide future research into braiding schemes in more exotic models.

Finally, the scheme offers new insight into the types of objects and braids that may yield interesting behaviour in more than two dimensions. For example, it is an instance of non-trivial braiding of point-like twists in a three dimensional code. This is interesting since point-like anyons in a three dimensional space are known to have trivial (bosonic or fermionic) statistics. It can, however, be understood as reflecting that the point-like twists are only endpoints to the one-dimensional domain wall, and that it is the rotation of this object that has a non-trivial action on the excitations present in it. Nonetheless, it highlights an interesting limitation to the conventional picture of twists as analogous to non-abelian anyons.

\subsection{Checkerboard Model}

The checkerboard model is a three dimensional stabiliser code that exhibits (foliated) fracton order \cite{Vijay,Shirley1,Dua}.  The excitations in this model have limited mobility, and unlike the previous stabilizer codes we have studied, the continuum limit of this checkerboard model is not described by a topological quantum field theory. Nonetheless, the model satisfies our definition of a topological stabilizer code, and so our results remain applicable to this class of codes. Indeed, as we now show, with a suitable encoding the Clifford group can be implemented by braiding twists in the checkerboard model. This is particularly interesting, as the complicated structure of logical operators and excitations in fracton models has previously limited the study of fault-tolerance with such codes.

The code is defined on a cubic lattice with qubits at the vertices.  We bi-colour the cubes with two colours such that all pairs of cubes that share a face are of different colours. The code has two types of stabilisers: products of $X$ operators and $Z$ operators respectively on all vertices of of one colour, as illustrated in Fig.~\ref{fig:CMStabs}.  No stabilizers are associated with cubes of the other colour. The code admits a range of excitations (see Ref.~\cite{Shirley1} for further details) but we will focus only on those used in our encoding.

\begin{figure}
\centering
\includegraphics{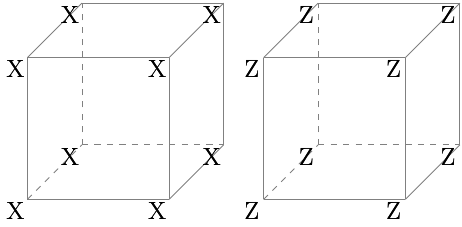}
\caption{The two types of stabilisers of the checkerboard model. Each stabiliser acts on every cube in the lattice of one colour. Cubes of the other colour do not have corresponding stabilisers. \label{fig:CMStabs}}
\end{figure}

In particular, we focus on planons within the model.   Planons are topological excitations that can only move freely within a plane of the code and give rise to additional excitations if they leave this plane \cite{Dua}. The planons we are interested in arise at the endpoints of adjacent, parallel strings of Pauli operators. Such strings commute with all stabiliers provided they form closed loops and are confined to a plane. If such a pair of strings instead terminate at endpoints (while still confined to a plane) then the result is a pair of point-like excitations at each endpoint. Individually, these excitations are lineons (excitations that can only freely move in a straight line), but taken as fused together the result is instead a planon \cite{Shirley2}. Such a planon consists of two particles in adjacent planes, which move together through their respective planes.

We focus our attention on these fused planons as our topological excitations of interest. By analogy to the toric code, we label these planons by $e$ if they correspond to strings of $Z$ operators and $m$ if they correspond to strings of $X$ operators. We indicate the planes to which the excitations are confined by subscripts. Thus, for example, excitation $e_{z=j,j-1}$ describes the planon at the end of a pair of $Z$ strings that are in the $xy$-planes of the lattice labelled by $z=j$ and $z=j-1$, respectively. 

We can consider braiding excitations that are confined to overlapping planes as though they were anyons in a two-dimensional model, as follows. Excitations $e_{z=j,j+1}$ and $m_{z=j,j+1}$ correspond to commuting strings of operators (since  where they cross they overlap on two qubits) and so braid trivially. However, $e_{z=j,j-1}$ and $m_{z=j,j+1}$ correspond to anticommuting strings (they overlap only on one qubit in the $z=j$ plane) and so yield a phase of $-1$ under braiding. Thus, this pair behaves as the $e$ and $m$ excitations of a two dimensional surface code when we consider braiding them as planons. The composite excitation $e_{z=j,j-1}m_{z=j,j+1}$ is a fermion with respect to exchange within the plane, again as with the $em$ fermion of the two dimensional surface code.

Symmetries of the model allow us to identify domains walls.  The checkerboard model is symmetric under exchanging $X$ and $Z$ operators on every qubit.  Assuming periodic boundary conditions, it is also symmetric under shifting all qubits one unit down and one unit back, e.g.~moving the qubit at position $(x,y,z)$ to $(x,y-1,z-1)$. The composition of these two symmetries is another symmetry, and applying this composite symmetry to a compact region of the code gives rise to a two-dimensional domain wall at the boundary. This domain wall will transform an excitation $m_{z=j+1,j}$ that crosses from outside the wall to inside to $e_{z=j,j-1}$. One dimensional twists at the endpoints of such a domain wall can condense a fermion $m_{z=j+1,j}e_{z=j,j-1}$.

We may now construct an encoding in these twists, and an implementation of Clifford logical gates, analogous to that of the two dimensional surface code construction in Sec.~\ref{2DSCTwist}. Specifically, we introduce a pair of domain walls in a pair of $yz$ planes separated in the $x$ direction by some distance $d$. We allow the walls to stretch around the whole code in the $y$ direction, but terminate them in a pair of one-dimensional line-like twists in the $z$ direction. We also ensure that these domain walls are oriented such that excitation $m_{z=j+1,j}$ maps to $e_{z=j,j-1}$ as it crosses from left to right for the wall on the left, but undergoes this transformation as it crosses from right to left through the wall on the right. We can now encode a logical qubit with logical operators analogously to in the two dimensional surface code, as shown in Fig.~\ref{CMEncoding}. Exchanging twists will now implement the single-qubit Clifford group analogously to in the two dimensional surface code. Specifically, exchanging the twists on the same domain wall will implement the $\bar{S}$ gate, while exchanging the front-left twist with the back-right will implement the $\bar{H}$ gates. We note that on a lattice of length $L_y$ in the $y$ direction, this scheme in fact allows for $2L_y$ logical qubits to be encoded in this twist setup, by considering all $L_y$ choices of $j$ and allowing for the role of $e$ and $m$ to be interchanged. Exchanging the twists implements the same Clifford operator simultaneously on all of these qubits.

\begin{figure}
\centering
\includegraphics{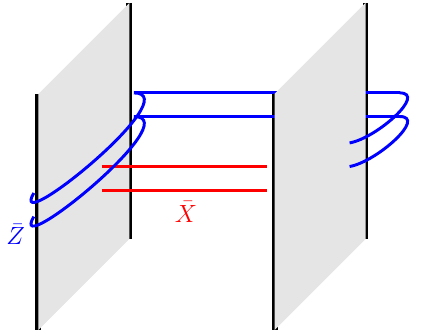}
\caption{The logical operators for a logical qubit encoded in the twists of the checkerboard model (with periodic boundary conditions). The twists are the black (topologically non-trivial) loops at the boundaries of the grey domain walls. Pairs of blue lines represent paths of the planon $m_{z=j+1,j}$ and pairs of red lines represent paths of the planon $e_{z=j,j-1}$.
\label{CMEncoding}}
\end{figure}

We emphasise that this scheme demonstrates the broad applicability of braiding defects, and of the approach we have taken in this section. Specifically, it shows that even fracton models, which have no analogue in two dimensions and do not admit descriptions in the context in which braiding is most familiar -- topological quantum field theories -- nonetheless may be studied using the tools and results we have discussed and developed. This offers an approach to developing schemes for fault-tolerant computation using such codes.

\subsection{Discussion of General Scheme}

The schemes presented in the previous subsections clarify a sufficient condition for a code to allow the full Clifford group by braiding. Specifically, this requirement is that the code admits domain walls with twists that can condense a generalised fermion (of any spatial dimension). For completeness, we now briefly summarise how a defect setup admitting the Clifford group can be constructed if such a requirement is satisfied abstractly.  A sketch of the two dimensional version of this scheme is provided in Fig.~\ref{GeneralClifford}.

\begin{figure}
\centering
\includegraphics{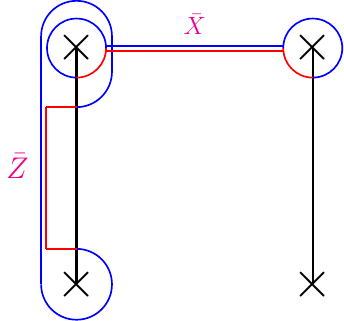}
\caption{ Illustration of two dimensional case of general scheme for encoding with twists on domain walls that can condense a fermion, $a$. Here the domain walls map $b$ (blue) $\to ab^*$ (red), so that $a$ is represented by blue and red lines together. The logical $\bar{X}$ operator corresponds to transferring $a$ between the domain walls, and $\bar{Z}$ corresponds to transferring $a$ between twists on the same domain wall. Exchanging twists in this setup, in analogy to Fig.~\ref{DWEncoding}, will implement the full single qubit Clifford group. \label{GeneralClifford}}
\end{figure}

We begin by noting that the significance of a $D$-dimensional code admitting twists that allow a generalised fermion to condense is that it implies that the corresponding domain wall interchanges two excitations that braid with one another to give a phase of $-1$. Indeed, a twist that allows generalised fermion $a$ to condense must act on some eigenstate excitation $b$ by the mapping $b \to ab^*$ (as discussed in Sec.~\ref{Sec:Twists}). Since $b$ and $ab^*$ must have the same phase under self-exchange (since they are interchanged by a domain wall \cite{Yoshida}) but the excitation $a$ produced by composing them is a generalised fermion, they must give a phase of $-1$ when braided around one another. This implies that the operators used to propagate them must anticommute \cite{Webster} and that the sum of the dimensions of excitations $a$ and $b$ must be $D-2$. This in turn implies that the dimension of the domain wall corresponding to the twist at which $a$ condenses must be twice the dimension of $a$ \cite{Webster}. These facts justify the validity of the following construction.

Assume a code admits some $k$-dimensional domain wall with twists that allows for some $j$-dimensional generalised fermion, $a$, to condense. Such a domain wall must act on an eigenstate excitation $b$ as $b \to ab^*$ to allow for excitation $a$ to condense. Construct a pair of such domain walls such that each has a pair of concentric twists topologically equivalent to $S^{j-1} \times S^{j-1}$ at their boundaries. Add a further puncture threaded through the holes of the domain walls if $j>1$ at which $a$ cannot condense. This prevents excitations $a$ that enclose the hole that condense on one of the domain walls from being absorbed by the vacuum. The logical $\bar{X}$ operator can then be specified by a process that propagates excitation $a$ between the pair of domain walls. This is equivalent up to stabilisers to a pair of excitations $b$ and $b^*$ condensing from the vacuum, excitation $b$ being propagated through the pair of domain walls, and then returning to its intial position and size without crossing through the walls again. This is the form in which the logical $\bar{X}$ operator was presented for the examples illustrated above. The logical $\bar{Z}$ operator can be realised by an operator that propagates excitation $b$ in a hypersurface that encloses one of the domain walls. This is equivalent to an operator that propagates the generalised fermion $a$ between the two twists of this domain wall, as shown in Fig.~\ref{GeneralClifford}. This choice of logical $\bar{X}$ and logical $\bar{Z}$ are valid since they anticommute, by the arguments of the previous paragraph. Exchanging two defects from the same domain wall will now implement a logical $\bar{S}$ operator, and exchanging the inner defect from one domain wall with the outer defect from the other will implement a logical $\bar{H}$ operator. These operators generate the full single qubit Clifford group. This can be extended to the full Clifford group on a larger set of qubits by adding additional holes and using these as ancilliary defects to allow for entangling gates, analogously to in \cite{Brown}.

\section{A Scheme for a Universal Fault-Tolerant Gate Set}
\label{Sec:Universal Scheme}

The fault-tolerant constructions we have presented in the previous section are constrained by the no-go theorem of Sec.~\ref{Sec:NoGo1}, and so do not yield universal gate sets for quantum computing.  We now turn to a scheme by which this no-go theorem can be circumvented, allowing for universal quantum computing with defects in a topological stabiliser code. Specifically, we present a scheme for realising a universal gate set on logical qubits encoded in defects in a stack of three 3D surface codes without magic state distillation. Our scheme is similar to that of Ref.~\cite{Vasmer}, but has the advantage that it can be performed entirely by manipulating defects in the bulk of the code and so allows for arbitrarily many logical qubits to be encoded in a single memory. In particular, we show that braiding of these defects can achieve the necessary entangling gates between these logical qubits for universality, without the need for lattice surgery between different code blocks. This illustrates how, despite the limitations discussed in Sec.~\ref{Sec:Limitations}, braiding defects in higher dimensional codes can still be a useful tool in the pursuit of universal fault-tolerant quantum computation.

The basis for this scheme is the use of stabiliser state injection to allow for a logical Hadamard operator to be implemented on qubits encoded in defects in a three-dimensional surface code. This injection circuit is not constrained by our no-go results, as it exploits adaptive implementation of a logical operator applied conditionally on the outcome of a logical measurement.  Such an adaptive process lies outside the scope of our no-go theorem, by supplementing locality-preserving and braiding logical operators with non-local classical processing \cite{BombinSingleShot} (see Appendix \ref{Appendix1} for further discussion). In particular, such operators could not otherwise be realised by braiding or as a locality-preserving logical operator. In addition to the scheme of Ref.~\cite{Vasmer}, the technique has also previously been applied to allow for universality in the 3D colour code \cite{BombinNoBraid} and a class of 2D Bacon-Shor codes \cite{YoderBacon}. To illustrate the key ideas of our scheme, we begin by discussing the case of three logical qubits, before then proceeding to the general case.

\subsection{Logical Qubits and Transversal Gates}
\label{Sec:LogQTrans}
We begin with a code that is locally equivalent to a stack of three 3D surface codes \cite{Webster}, labelled one, two and three. Three logical qubits can be encoded in five defects in this code, as illustrated in Fig.~\ref{UniDefects}. Specifically, this setup consists of four toroidal defects arranged with a fifth toroidal defect threaded through the hole of each other defect. Appropriately choosing boundaries of these defects to be rough and smooth with respect to the three surface codes allows for three logical qubits to be encoded, with logical Pauli operators as illustrated in Fig.~\ref{Fig15}. Further details of these defects are provided in Appendix \ref{Appendix2}. 

This setup inherits locality-preserving logical operators (in fact, transversal ones) from the underlying three dimensional surface codes. Along with logical Pauli operators, it is known that such codes (with boundaries) admit transversal $\overline{CZ}$ and $\overline{CCZ}$ logical operators \cite{KubicaUnfold,Vasmer}. This ensures that the stabiliser group of the code is preserved by such operators in our setup (noting that altered stabilisers at the defects in our setup correspond to stabilisers on code boundaries). Thus, we can leverage such transversal operators of the underlying code into transversal operators in our setup.

Specifically, these logical qubits admit transversal $\overline{\text{CZ}}_{ij}$ operators, implemented by applying $\text{CZ}_{ij}$ transversally across the support of $\bar{X}_k$ (for $\{i,j,k\}=\{1,2,3\}$). Transversally applying $\text{CCZ}_{123}$ across the whole defect setup implements the logical operator $\overline{\text{CCZ}}_{123}$. We note that these transversal operators are all analogous to the transversal operators implementable across three-dimensional surface codes with boundaries used for encoding instead of defects.

\begin{figure}
\centering
\includegraphics{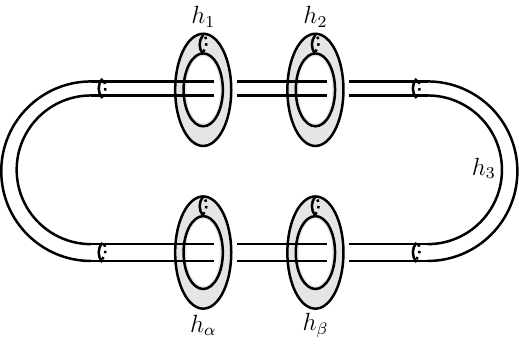}
\caption{Defect setup for the universal scheme we present. All five defects are punctures in a three dimensional colour code (equivalent to three surface codes). $h_1$ and $h_2$ are rough with respect to surface code one and smooth with respect to surface codes two and three. $h_\alpha$ and $h_\beta$ are rough with respect to surface code two and smooth with respect to surface codes one and three. $h_3$ is threaded through all four other holes, and is rough with respect to surface codes one and two and smooth with respect to surface code three. \label{UniDefects}}
\end{figure}

\begin{figure}
\centering
\subfigure[Pauli operators for logical qubit $1$ and ancilla qubit $a$. These are the two qubits encoded using excitations from code 1. Specifically, blue here is a path of $m_1$ and red is a path of $e_1$. $\bar{X}_1$ is a torus around $h_1$ and $\bar{Z}_1$ is a line between $h_1$ and $h_2$. $\bar{X}_a$ is a cylinder between $h_4$ and $h_5$ and $\bar{Z}_a$ is a loop around $h_4$. \label{Q1a}]{
\includegraphics[scale=.9]{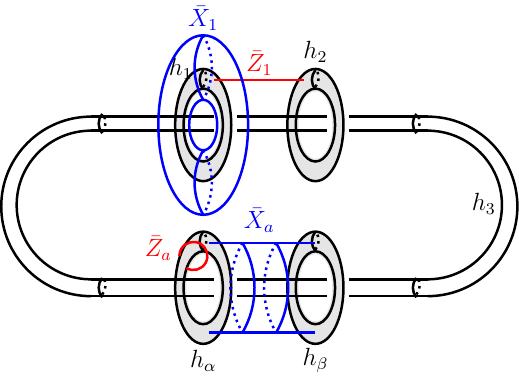}}
\subfigure[Pauli operators for logical qubit $2$ and ancilla qubit $b$. These are the two qubits encoded using excitations from code 2. Specifically, blue here is a path of $m_2$ and red is a path of $e_2$. $\bar{X}_2$ is a cylinder between $h_1$ and $h_2$ and $\bar{Z}_2$ is a loop around $h_1$. $\bar{X}_b$ is a torus around $h_4$ and $\bar{Z}_b$ is a line between $h_\alpha$ and $h_\beta$. \label{Q2b}]{
\includegraphics[scale=.9]{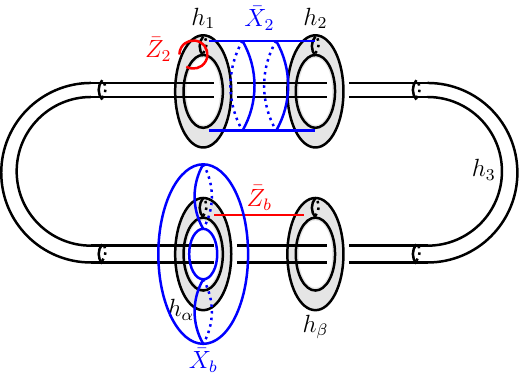}}
\subfigure[Pauli operators for logical qubit $3$. This is the qubit encoded using excitations from code 3. Specifically, blue here is a path of $m_3$ and red is a path of $e_3$. $\bar{X}_3$ is a disk stretching out to $h_3$ and $\bar{Z}_3$ is a loop around $h_3$. \label{Q3}]{
\includegraphics[scale=.9]{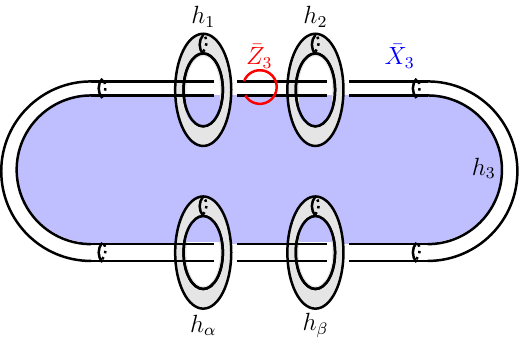}}
\caption{Logical Pauli operators of the encoded qubits.\label{Fig15}}
\end{figure}

\subsection{Logical Hadamard Operator by State Injection}
\label{Sec:LogH}

To complete a universal gate set on three logical qubits, we now need only to supplement the above transversal gates with a logical Hadamard operator for each logical qubit. This can be done using stabiliser state injection, using the circuit shown in Fig.~\ref{Fig:InjectionCircuit}. However, there are a number of issues that must be addressed to allow for this circuit to be used to realise logical Hadamard operators on each logical qubit. These are addressed in detail in Appendix \ref{Appendix3}; here we  summarise the issues and their solutions.

First, we require ancillary logical qubits prepared in the logical $|\bar{+}\rangle$ state. Fig.~\ref{Q1a} and Fig.~\ref{Q2b} show how two ancilla qubits, $a$ and $b$, can be encoded. Either of these ancillae can be prepared in state $|\bar{+}\rangle$ by measuring their logical $\bar{X}$ operators and applying Pauli corrections if necessary. Logical $\overline{\text{CZ}}_{3a}$ and $\overline{\text{CZ}}_{3b}$ operators can be implemented transversally by implementing $\text{CZ}_{31}$ on the support of $\bar{X}_b$ and $\text{CZ}_{23}$ on the support of $\bar{X}_a$ respectively. This enable a logical $\bar{H}_3$ operator to be applied using these ancillae, as described in Appendix \ref{Appendix3e}.

Applying $\bar{H}_1$ and $\bar{H}_2$ is more challenging, because transversal $\overline{\text{CZ}}$ operators cannot be realised transversally between logical qubits one and two and the ancillae. This problem is overcome by braiding. Specifically, braiding defect $h_1$ around $h_\alpha$ implements the logical operator $\overline{\text{CNOT}}_{a1}\overline{\text{CNOT}}_{2b}$. This can be seen to be analogous to braiding holes in the two dimensional surface code, which implements entangling $\overline{\text{CNOT}}$ operators between encoded logical qubits. This entangling operator allows for logical information to be transferred between the logical qubits 1 and 2 and the encoded ancilla qubits, which then allows for the circuit necessary for logical $\bar{H}_1$ and $\bar{H}_2$ to be implemented using transversal $\overline{\text{CZ}}$ operators. This is described in more detail in Appendix \ref{Appendix3e}. These logical $\bar{H}$ operators combined with the transversal $\overline{\text{CCZ}}$ operator suffice for a universal gate set.

\begin{figure}
\centering
\includegraphics{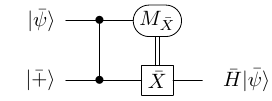}
\caption{A state injection circuit allowing for a logical Hadamard operator to be implemented using only logical $\overline{\text{CZ}}$ and $\bar{X}$ conditional on logical $X$ measurement. The double line indicates a classical control. Specifically, $\bar{X}$ is applied if and only if the outcome of the logical measurement is $-1$. \label{Fig:InjectionCircuit}}
\end{figure}

\subsection{Universal Gate Set on $N$ Logical Qubits}
Having established how a universal gate set may be implemented on three logical qubits, we now consider a larger number of logical qubits. In particular, we show that $N$ logical qubits admitting a universal gate set may be encoded in $N+2$ defects (for odd $N\geq 3$). Specifically, we maintain a single threaded defect through the centres of $\frac{N+1}{2}$ pairs of holes. In one pair of holes we encode a pair of ancillae as in the three logical qubit case, while we encode a pair of logical qubits in each other pair of defects and one logical qubit using the threaded defect analogously to qubit three in the previous case. 

The universal gate set on this set of logical qubits again emerges from combining transversal $\overline{\text{CCZ}}$ logical operators with Hadamard logical operators implemented by state injection. The full details of how this can be done are presented in Appendix \ref{Appendix4}. We here consider only at a high level the two major issues that emerge in the general case that did not in the three qubit case.

First, universality requires that we can implement $\overline{\text{CCZ}}$ between any three logical qubits; even those associated with different defect pairs. This is similar to the problem of implementing a $\overline{\text{CCZ}}$ operator between logical qubits in different code blocks in the scheme of Vasmer and Browne \cite{Vasmer}. They address this using lattice surgery, which allows for logical qubits to be swapped between different code blocks. An advantage of our scheme based on defects is that braiding may be used to swap logical qubits between different defect pairs, using $\overline{\text{CNOT}}$ operators. This allows for logical qubits to be swapped around to bring them into a configuration that admits a $\overline{\text{CCZ}}$ operator; without requiring separate code blocks of lattice surgery.

Second, the presence of many logical qubits in the same code block adds a challenge in implementing transversal gates. In particular, the transversal $\overline{\text{CCZ}}$ operator described in Sec.~\ref{Sec:LogQTrans} has support on all qubits across the whole defect setup. However, applying such an operator in the presence of many logical qubits would implement all possible $\overline{\text{CCZ}}$ operators simultaneously. Specifically, it would entangle the logical qubits in each defect pair with the logical qubit encoded using the threaded defect, without being able to select out the logical qubits of a particular pair. This can be overcome by isolating the pair of logical qubits associated out from the other logical qubits temporarily to allow for a $\overline{\text{CCZ}}$ operator to be implemented on this pair and the logical qubit associated with the threaded defect, without involving other logical qubits. This can be done by braiding defects to implement $\overline{\text{CNOT}}$ operators between the qubits to be isolated and the encoded ancillae. Once logical qubits are stored in a defect pair along with an ancilla, the ancilla may be used to control whether or not the logical qubit is affected by the implementation of the transversal $\overline{\text{CCZ}}$ operator This allows for $\overline{\text{CCZ}}$ operators to be implemented on just three logical qubits, and not across all logical qubits, as required.

\subsection{Discussion of this approach to universality}

We emphasise that this universal scheme has the advantage over similar alternative schemes, such as that of Ref.~\cite{Vasmer}, in that it provides for a universal gate set on an arbitrary number of logical qubits within a single quantum memory.  As a result, we do not require a large number of code blocks, including ancilla blocks, to be entangled using lattice surgery. This advantage has its origins in the use of braiding to entangle logical qubits associated with different defect pairs. Such braiding is made possible only by our consideration of encoding logical qubits in defects. In particular, it plays the role it does in our universal scheme only because we have considered braiding in higher dimensional codes.  That is, it is by moving beyond two dimensions that allows for non-Clifford locality-preserving logical operators, which can be supplemented by the entangling Clifford gates by braiding as we have presented.

Our approach also has the advantage that it does not require additional encoded ancilla qubits as the number of logical qubits increases. This efficiency is made possible by our technique of shifting logical information onto ancilla qubits when necessary to allow for desired gates to be implemented only on specific logical qubits. This technique could be more broadly useful for other schemes that have many logical qubits encoded in the same memory, that often face the problem of transversal gates not being specific to desired logical qubits.

\section{Conclusion}

In this paper, we have explored both the limitations and potential for fault-tolerant quantum computing by braiding defects in topological stabiliser codes. In particular, we have shown that universal quantum computing cannot be achieved by braiding alone, nor by supplementing braiding logical operators by locality-preserving logical operators. This extends the result of Bravyi and K{\"onig}~\cite{Bravyi} to a far larger class of fault-tolerant operators. Indeed, we have also noted that this constraint applies more generally to abelian quantum double models, extending the results of Ref.~\cite{Webster}.

Our result has implications for realising universal fault-tolerant quantum computing. In particular, it proves that schemes based on braiding defects in topological stabiliser codes or abelian quantum double models must be supplemented by other techniques to permit universality. The most widely-employed approach to this is magic state distillation, which can lift schemes allowing for the full Clifford group to universality \cite{Raussendorf1,Raussendorf2}. We have furthered the potential for this approach by demonstrating that the Clifford group may be realised by braiding in a large class of topological stabiliser codes, specifically, any code with topological defects that can condense a generalised fermion.  While progress has been made towards reducing the prohibitive overhead of magic state distillation \cite{Litinski}, this cost remains a significant contributor to the large number of physical qubits required for fault-tolerant quantum computing. Our work further emphasises the importance of continuing work to improve on this overhead, or to seek alternative approaches. 

Beyond magic state distillation, other approaches to universality can also be considered, including those that make more explicit use of topological protection. We have presented such an approach by adapting the scheme of Vasmer and Browne  based on stabiliser state injection to develop a universal gate set on logical qubits encoded in defects \cite{Vasmer}. This scheme uses a combination of locality-preserving and braiding logical operators, but supplements these operators with non-local classical processing to circumvent our no-go results. An advantage of this scheme is that arbitrarily many logical qubits can be encoded into a single code block, removing the need for lattice surgery. It would be valuable to consider whether this scheme may offer an improved overhead compared with Ref.~\cite{Vasmer}. This scheme could also be used as a model for adapting existing approaches to universality on qubits encoded in boundaries to qubits encoded in defects. In particular, future work could seek to adapt other techniques used to achieve universality such as dimensional jumping \cite{BombinGCC} and just-in-time gauge fixing \cite{Bombin2018,BrownJIT} to defect encodings. Future work could also explore and seek to develop approaches to universality that are specifically designed for topological defect schemes, such as topological charge measurement \cite{Cong}. The potential of all these different approaches should also be assessed by investigating the overhead of each technique to determine if they offer an advantage over the conventional alternative of using the two dimensional surface code with magic state distillation.

\begin{acknowledgments}
We acknowledge discussions with Benjamin Brown, Markus Kesselring, Sam Roberts, Michael Vasmer and Dominic Williamson. This work is supported by the Australian Research Council (ARC) via the Centre of Excellence in Engineered Quantum Systems (EQuS) project number CE170100009 and Discovery Project number DP170103073. PW also acknowledges support from The University of Sydney Nano Institute via The John Makepeace Bennett Gift.
\end{acknowledgments}

\appendix
\renewcommand{\thesection}{}
\section{}  
\label{Appendix}
In this appendix, we provide a complete presentation of the construction summarised in Sec.~\ref{Sec:Universal Scheme}. Specifically, we begin by discussing the background to this approach to universality, and how it violates the assumptions of the no-go theorem of Sec.~\ref{Sec:NoGo1}. We then present a scheme using punctures in a code equivalent to a stack of three 3D surface codes that allows for this approach. We go on to show how that scheme admits a universal gate set on three logical qubits generated by Hadamard operators on each qubit and a $\text{CCZ}$ between the three logical qubits. Finally, we explain how this scheme may be extended to allow for a universal gate set on  $N$ logical qubits in a setup with a total of $N+2$ punctures (for any odd $N\geq 3$). 

\subsection{Approach to Circumventing the No-Go Theorem}
\label{Appendix1}
Theorem \ref{Theorem1}, which shows that a universal gate set cannot be realised using only braiding and locality-preserving logical operators, can be circumvented by allowing for non-local classical processing. Specifically, the scheme we present uses state injection to achieve logical Hadamard operators in topological stabiliser codes that do not admit them as locality-preserving logical operators. Such a technique allows for universality from only the ability to prepare both logical $\bar{X}$ and $\bar{Z}$ eigenstates, perform logical $\bar{X}$ and $\bar{Z}$ measurements and implement locality-preserving $\overline{CCZ}$ operators, provided that logical operators may be applied conditionally on logical measurement outcomes \cite{YoderBacon}. Our scheme applies this technique to a setup in which logical qubits are encoded in defects, using a combination of locality-preserving and braiding logical operators.

We note that the reliance of this scheme on applying logical operators conditionally on logical measurement outcomes makes it dependent on non-local classical communication. Specifically, information from the measurements of physical qubits around the non-local support of logical Pauli operators must be combined to determine the measurement outcome. The result of this combination of non-local information is then manifested in the application of logical operators on the quantum system. The scheme is thus not fundamentally local, since the application of operators on physical qubits in one part of the code can be affected by the measurement outcomes on distant physical qubits. However, since all operations on the quantum memory are local, it is nonetheless fault-tolerant, provided the classical computer in which the measurement outcomes are combined and processed is sufficiently reliable that local classical errors that could be propagated to non-local quantum errors are highly unlikely. This approach to circumventing Theorem \ref{Theorem1} is similar to approaches taken to circumventing similar no-go theorems such as that of Eastin-Knill and Bravyi-Konig \cite{BombinGCC}.

\subsection{Defects and Encoding}
\label{Appendix2}
We now describe the defects and encoding that allow us to realise a universal gate set. This setup can be realised in a code locally equivalent to three copies of a 3D surface code. We refer to these as codes 1, 2 and 3.

We begin by considering a pair of toroidal punctures. We choose these punctures to be rough with respect to surface code 1 and smooth with respect to surface codes 2 and 3. This means that they allow for $e_1$, $m_2$ and $m_3$ to condense, but not $m_1$, $e_2$ or $e_3$. We note that $s_{ij}$ excitations can condense at a boundary if and only if the boundary is rough with respect to code $i$ or $j$ \cite{Webster}. Thus, these punctures allow for $s_{12}$ and $s_{31}$ to condense, but not $s_{23}$. We label the two punctures $h_1$ and $h_2$.

We introduce a further toroidal puncture threaded through these two punctures. We choose this puncture to be rough with respect to codes 1 and 2 and smooth with respect to code 3. This means it allows $e_1$, $e_2$, $m_3$, $s_{12}$, $s_{23}$ and $s_{31}$ to condense, but not $m_1$, $m_2$ or $e_3$. This threaded defect thus prevents $m_2$ that condense at one of the first two punctures from being absorbed into the vacuum. We label this threaded defect $h_3$. The defect setup is presented in Fig.~\ref{UniDefects}.

These three defects allow for three logical qubits to be encoded. We define these logical qubits implicitly by specifying their logical Pauli operators. These are shown in Fig. \ref{Fig15}. Each of the three logical qubits are implemented by physical operators acting on the corresponding surface code. Specifically, qubit one has logical $\bar{X}_1$ operator realised by a torus of physical $X_1$ operators acting around defect $h_1$. It has logical $\bar{Z}_1$ operator implemented by a line of physical $Z_1$ operators between $h_1$ and $h_2$. Qubit two has $\bar{X}_2$ operator that is a cylinder of physical $X_2$ operators between $h_1$ and $h_2$ and $\bar{Z}_2$ operator that is a loop of $Z_2$ operators around $h_1$. Qubit three has $\bar{X}_3$ operator that is a membrane of physical $X_3$ operators terminating in a loop on defect $h_3$ and $\bar{Z}_3$ operator that is a loop of $Z_3$ operators around $h_3$. 

We can verify that these operators define a valid encoding of three qubits by checking commutation relations of these logical operators. Specifically, logical operators acting on different logical qubits commute since they consist of physical operators acting on different surface codes. Logical $\bar{X}_i$ and $\bar{Z}_i$ operators consist of physical $X$ and $Z$ operators respectively acting on the same surface code, and intersect at a point, and so anticommute.

\subsection{Universal Gate Set for Three Qubits}
\label{Appendix3}
We now describe how to realise universal quantum computing on three logical qubits encoded in these defects. Specifically, we first show how transversal $\overline{\text{CZ}}$ operators may be implemented pairwise between logical qubits and then how transversal $\overline{\text{CCZ}}$ may be implemented between all three. We then demonstrate that a logical Hadamard operator may be implemented on any of the logical qubits with the addition of two ancilla qubits encoded in an additional pair of defects. Together, this gives a universal gate set.

\subsubsection{Transversal Gates}
\label{Appendix3a}
Our encoding allows for transversal $\overline{\text{CZ}}$ operators between each pair of logical qubits and a transversal $\overline{\text{CCZ}}$ between all three qubits. Specifically, we may realise a transversal $\overline{\text{CZ}}_{ij}$ operator by applying $\text{CZ}_{ij}$ transversally across the support of $\bar{X}_k$, as illustrated in Fig.~\ref{CZ}. To see this, observe from Fig.~\ref{Fig15} that the intersection of the support of $\overline{\text{CZ}}_{ij}$ and the support of $\bar{X}_i$ is the support of $\bar{Z}_j$. We note that this situation is analogous to that in the 3D colour code with cubic boundary conditions, that similarly encodes three logical qubits and admits $\overline{\text{CZ}}$ operators between each pair \cite{KubicaUnfold}. Since the support of $\bar{X}_i$ is the same as that of $\overline{\text{CZ}}_{jk}$ for all $i,j,k$, we may implement $\overline{\text{CCZ}}_{ijk}$ transversally by applying $\text{CCZ}_{ijk}$ at each site of the code. This again, is analogous to the 3D colour code with cubic boundary conditions that again allows for a transversal $\overline{\text{CCZ}}$ operator.

\begin{figure}
\centering
\includegraphics{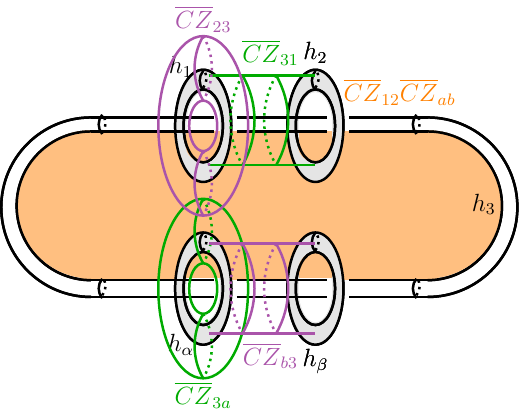}
\caption{Transversal logical $\overline{\text{CZ}}$ operators acting on the encoded qubits. Purple shows path of $s_{23}$, green shows path of $s_{31}$ and orange shows path of $s_{12}$. $\overline{\text{CZ}}_{12}\overline{\text{CZ}}_{ab}$ can be implemented by growing $s_{12}$ from the vacuum out to defect $h_3$. The propagation of $s_{ij}$ is realised by transversal implementation of $\text{CZ}_{ij}$. $\overline{\text{CZ}}_{13}$ may be implemented by a cylinder propagating $s_{31}$ between $h_1$ and $h_2$. $\overline{\text{CZ}}_{b3}$ may be implemented by a cylinder propagating $s_{23}$ between $h_\alpha$ and $h_\beta$. $\overline{\text{CZ}}_{23}$ may be implemented by a torus taking $s_{23}$ around $h_1$. $\overline{\text{CZ}}_{a3}$ may be implemented by a torus taking $s_{31}$ around $h_\alpha$.\label{CZ}}
\end{figure}

\subsubsection{Entangling with Ancilla Qubits}
\label{Appendix3b}
We now turn our attention to how we may supplement the transversal $\overline{\text{CCZ}}$ operator with fault-tolerant Hadamard operators on each qubit to give universality. This is done using a state injection circuit requiring ancilla qubits. We first introduce how these ancilla qubits are encoded and how they may be entangled with the data qubits by both braiding and transversal operators. Finally, we describe how the circuit in Fig.~\ref{Fig:InjectionCircuit} may be applied to take a state, $|\bar{\psi}\rangle$ on a data qubit to $\bar{H}|\bar{\psi}\rangle$ on an ancilla qubit. We then present a full procedure for implementing a Hadamard on a logical qubit such that the final state $\bar{H}|\bar{\psi}\rangle$ is on the original logical qubit, rather than the ancilla.

To define the ancilla qubits, we introduce two additional defects similar to holes $h_1$ and $h_2$ but with different boundaries. Specifically, we choose these to be toroidal punctures, labelled $h_\alpha$ and $h_\beta$, such that $h_3$ is threaded through their holes similarly to for $h_1$ and $h_2$. We choose them both to be rough with respect to code 2 and smooth with respect to codes 1 and 3. That is, they allow for $m_1, e_2, m_3, s_{12}$ and $s_{23}$ to condense. We note that this is similar to punctures $h_1$ and $h_2$ but with boundaries for codes one and two that are swapped. We then define two ancilla qubits, labelled $a$ and $b$, similarly to logical qubits $1$ and $2$ (but with codes one and two interchanged). This is shown in Fig.~\ref{Q1a} and Fig.~\ref{Q2b}.

Having introduced additional encoded qubits, we now revisit the transversal gates. Firstly, we note that we have transversal operators $\overline{\text{CZ}}_{3a}$ and $\overline{\text{CZ}}_{3b}$ analogously to $\overline{\text{CZ}}_{31}$ and $\overline{\text{CZ}}_{23}$, as shown in Fig.~\ref{CZ}. This allows us to entangle qubit 3 with the ancilla qubits. We also note that the $\overline{\text{CZ}}_{12}$ operator identified previously acts as $\overline{\text{CZ}}_{12}\overline{\text{CZ}}_{ab}$ now that we have ancilla qubits present. We can ensure that this operator acts only as $\overline{\text{CZ}}_{12}$ however by preparing one of the ancilla qubits in state $|\bar{0}\rangle$ so that $\overline{\text{CZ}}_{ab}$ acts trivially.

The introduction of additional defects including ancilla qubits also allows for braiding logical qubits that entangle data qubits with ancilla qubits. Specifically, we note that braiding $h_2$ around $h_\alpha$ acts as $\overline{\text{CNOT}}_{a1}\overline{\text{CNOT}}_{2b}$. To see this, note that it is analogous to the braiding of holes in the two dimensional surface code. Specifically, here, when a defect is braided around another, a cylindrical operator that terminates on one of the holes being braided will have a torus around the other braided defect appended to it. Thus, considering Fig.~\ref{Q1a} and Fig.~\ref{Q2b}, the braiding operator maps $\bar{X}_2 \to \bar{X}_2\bar{X}_b$ and $\bar{X}_a \to \bar{X}_a\bar{X}_1$. Also, a line that terminates on one of the holes being braided will have a loop around the other braided defect appended to it. Thus, again considering Fig.~\ref{Q1a} and Fig.~\ref{Q2b}, the braiding operator maps $\bar{Z}_1 \to \bar{Z}_1\bar{Z}_a$ and $\bar{Z}_b \to \bar{Z}_b\bar{Z}_2$. We note that we can also implement $\overline{\text{CNOT}}_{\alpha 1}$ and $\overline{\text{CNOT}}_{2\beta}$ individually by appropriately preparing ancilla qubits to make the other part of the operator act trivially.

\subsubsection{Logical Measurements and Preparation}
\label{Appendix3c}
In addition to the transversal and braiding logical operators we have identified, we can also implement logical measurements in the $\bar{Z}$ and $\bar{X}$ eigenbases for any of the encoded qubits. This can be done by leveraging the transversal $\bar{X}$ and $\bar{Z}$ operators. Specifically, we can perform a (non-fault-tolerant) logical $\bar{Z}_i$ measurement by measuring $Z_i$ on each qubit on a path that supports a logical $\bar{Z}_i$ operator (or similarly for $\bar{X}_i$). This can be done fault-tolerantly by using $d$ (physical) ancilla qubits for each (physical) code qubit, where $d$ is the weight of the logical operator. This is shown in Fig.~\ref{MZ} for the case where $d=3$. We note that these logical measurements allow us to fault-tolerantly prepare states $|\bar{0}\rangle$ and $|\bar{+}\rangle$ by performing measurements in the $\bar{Z}$ and $\bar{X}$ bases respectively and performing necessary Pauli corrections \cite{Vasmer}.

\begin{figure}
\centering
\includegraphics{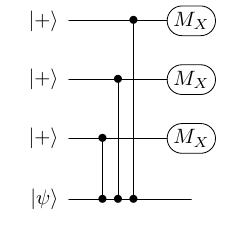}
\caption{Fault-tolerant circuit for measurement of ${Z}$ on the qubit (bottom wire) initially in state $|\psi\rangle$ in the case where the code distance is $d=3$. Logical measurement of $\bar{Z}$ on an encoded qubit can be performed by performing this measurement circuit on each physical qubit on a region that supports a logical $\bar{Z}$ operator. Fault-tolerant measurement of  ${X}$ can be performed with the same circuit but with each $\text{CZ}$ operators replaced with a $\text{CNOT}$ operator whose target is the bottom qubit. \label{MZ}}
\end{figure}

\subsubsection{Gadgets}
\label{Appendix3d}
We now define two gadgets that we can use as tools to realise a Hadamard logical operator. The first of these gadgets is that described in Sec.~\ref{Sec:LogH} and shown in Fig.~\ref{Gadget Hxy}. We call this gadget $H_{xy}$. Here, $x$ is initially a data qubit in logical state $|\bar{\psi}\rangle$ that ends up encoding no information after the gadget is applied. Qubit $y$ is initially an ancilla qubit that ends up storing logical information and in the state $\bar{H}|\bar{\psi}\rangle$.

The other gadget that we make use of we label $I_{xy}$ and show in Fig.~\ref{Ixy}. This gadget takes logical information on data qubit $x$ and moves it to an ancilla qubit $y$. We note that we cannot implement unitary logical SWAP operators between data qubits $1$ or $2$ and ancilla qubits $a$ or $b$ since the CNOT operators that we can implement between them by braiding can only act in one direction (i.e.~the control and target logical qubits cannot be swapped). Nonetheless, by using this gadget we begin with qubit $x$ in state $|\bar{\psi}\rangle$ and end with $y$ in this same state. Qubit $x$ ends up encoding no logical information.

\begin{figure}
\centering
\includegraphics{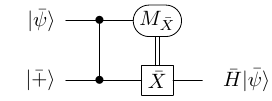}
\caption{The gadget $H_{xy}$ which takes a logical state $|\bar{\psi}\rangle$ onto the state $\bar{H}|\bar{\psi}\rangle$ on an ancilla qubit. Note that this gadget is the same circuit as presented in the text in Fig.~\ref{Fig:InjectionCircuit}. \label{Gadget Hxy}}
\end{figure}

\begin{figure}
\centering
\includegraphics{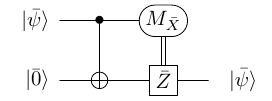}
\includegraphics{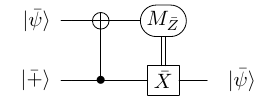}
\caption{The gadget $I_{xy}$ which takes a logical state $|\bar{\psi}\rangle$ onto an ancilla qubit. It can be implemented by either of the above circuits, depending on which direction of $\overline{\text{CNOT}}$ gate is available for the encoded qubits, $x$ and $y$. \label{Ixy}}
\end{figure}

\subsubsection{Logical Hadamard Operators}
\label{Appendix3e}
We now describe how logical Hadamard operators can be implemented on each of the three logical qubits to complete a universal gate set on these qubits. We consider first logical qubit three, since it allows the easiest implementation, and then describe the more complicated approach taken for qubits one and two.

To implement $\bar{H}_3$, we can implement the gadget $H_{xy}$ three times in a cycle then begins and ends on qubit three. In this way, the state, $|\bar{\psi}\rangle$, on qubit three will have $\bar{H}^3=\bar{H}$ applied to it and end the process back on qubit three. In particular, we can do this with ancilla qubits $a$ and $b$ by applying $H_{b3} H_{ab} H_{3a}$. This is possible since $\overline{\text{CZ}}_{b3}$ and $\overline{\text{CZ}}_{3a}$ are both implementable transversally. We can implement $\overline{\text{CZ}}_{ab}$ by first preparing qubit $b$ in state $|\bar{0}\rangle$, applying the transversal operator $\overline{\text{CZ}}_{12}\overline{\text{CZ}}_{ab}$, then preparing $b$ in $|\bar{+}\rangle$ and applying $\overline{\text{CZ}}_{12}\overline{\text{CZ}}_{ab}$ again. This has the effect of implementing $\overline{\text{CZ}}_{12}$ twice and so does not affect these qubits, but does implement $\overline{\text{CZ}}_{ab}$ on the ancilla qubits as necessary since the first implementation of $\overline{\text{CZ}}_{12}\overline{\text{CZ}}_{ab}$ acts trivially on these qubits. We can diagrammatically represent the process used to implement $\overline{H}_3$ in the way shown in Fig.~\ref{H3Diagram}. We can also draw the circuit used explicitly as shown in Fig.~\ref{H3Circuit}.

\begin{figure}
\centering
\includegraphics{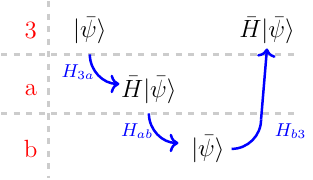}
\caption{Diagrammatic representation of how the sequence of gadgets $H_{b3} H_{ab} H_{3a}$ can be used to implement the logical operator $\bar{H}_3$.\label{H3Diagram}}
\end{figure}

\begin{figure}
\centering
\includegraphics[width=0.48\textwidth]{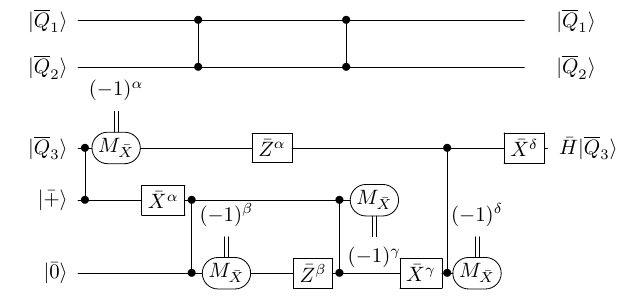}
\caption{Explicit circuit that can be used to implement $\bar{H}_3$. Double lines show the measurement outcome of logical measurements. The parameters specifying these outcomes are then used to control future logical operators. \label{H3Circuit}}
\end{figure}

For logical qubits 1 and 2 we are more limited in the entangling gates we can perform than for logical qubit 3. This means we must use a more complicated scheme to realise $\bar{H}_1$ and $\bar{H}_2$ than for $\bar{H}_3$. In particular, we can implement $\bar{H}_1$ by applying $H_{a3}H_{31}H_{b3}H_{3a}H_{ab}I_{1a}$. Similarly, we can implement $\bar{H}_2$ by applying $H_{b3}H_{32}H_{a3}H_{3b}H_{ab}I_{2b}$. Verifying that these indeed implement the required logical Hadamard operators is best done by considering them diagrammatically, as presented in Fig.~\ref{H1Diagram} and Fig.~\ref{H2Diagram}. The explicit circuits to implement this can be realised by combining together the circuits for each gadget.

\begin{figure}
\centering
\includegraphics{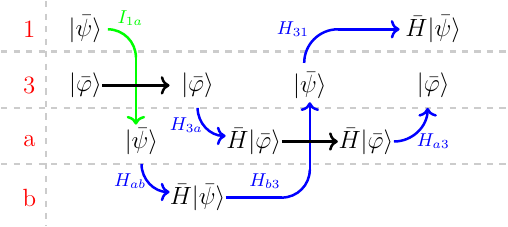}
\caption{Diagrammatic representation of how the sequence of gadgets $H_{a3}H_{31}H_{b3}H_{3a}H_{ab}I_{1a}$ can be used to implement the logical operator $\bar{H}_1$.\label{H1Diagram}}
\end{figure}

\begin{figure}
\centering
\includegraphics{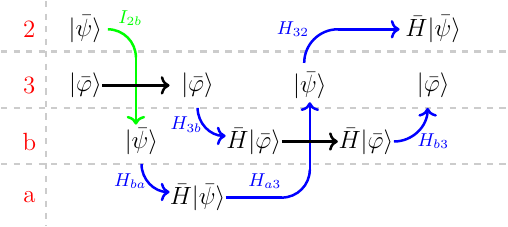}
\caption{Diagrammatic representation of how the sequence of gadgets $H_{b3}H_{32}H_{a3}H_{3b}H_{ab}I_{2b}$ can be used to implement the logical operator $\bar{H}_2$.\label{H2Diagram}}
\end{figure}

We note two interesting features of the circuits used to implement the logical Hadamard operators on logical qubits 1 and 2. Firstly, we note that both $\bar{H}_1$ and $\bar{H}_2$ rely on braiding defects, since they use the gadget $I_{xy}$ which involves performing a $\text{CNOT}$ operator by braiding. We also note that $\bar{H}_1$ and $\bar{H}_2$ both use a technique of shifting logical information around between encoded qubits as a tool to allow for the right gate to be implemented. In particular, the gadgets $H_{a3}$ and $H_{3a}$ in $\bar{H}_1$ (and similar for $\bar{H}_2$) are used only to move logical information from logical qubit 3 onto an ancilla qubit. This allows for logical qubit 3 to be used as an ancilla, while protecting the logical information of the qubit on qubit $a$. Such shifting of logical information around plays an important role here, and will also do so in the following section.
 
\subsection{Universal Gate Set for $N$ Qubits}
\label{Appendix4}
Having demonstrated that a universal gate set on three logical qubits can be realised with this defect setup, we now show that this may be generalised to give universality on an arbitrarily large number of logical qubits. This procedure uses only locality-preserving logical operators and braiding supplemented by logical measurement and classical processing. In particular, it does not require additional techniques such as lattice surgery to transfer information between different code blocks, or non-local operators.

We add more logical qubits to the scheme by adding more pairs of defects identical to $h_1$ and $h_2$, which also have the threaded defect threaded through them. Each new defect pair then adds two more logical qubits defined analogously to logical qubits one and two, but using this new pair of defects. In this way, we can construct an encoding of $N$ logical qubits in a setup with a total of $N+2$ punctures (for odd $N\geq 3$). We note that we will label the qubit known as qubit three in the three-qubit case (encoded using the threaded defect) by qubit $N$, and label the pairs of logical qubits encoded in the $k$th pair of defects by qubit $2k-1$ and $2k$. Note also that we may implement the gadgets $I_{2K-1,a}$ and $I_{2K,a}$ for any $k$. Indeed, braiding one of the punctures in the $k$th pair around hole $\alpha$ with qubit $b$ prepared in state $|\bar{+}\rangle$ allows the necessary entangling gate, and logical measurement and Pauli operators can be applied as before.

To justify that we have a universal gate set on this full set of logical qubits, we must now show that our scheme allows for fault-tolerant Hadamard logical operators on each logical qubit, and $\overline{\text{CCZ}}$ operators on each subset of three logical qubits.

We consider $\overline{\text{CCZ}}$ first. Specifically, with $N+2$ punctures, we now have that transversal CCZ acts act as in Eq.~\ref{A1} (where $\overline{\text{CCZ}}_{a,b,c}$ indicates that $CCZ$ acts on qubits labelled $a$, $b$ and $c$).
\begin{equation}\label{A1}
\overline{\text{CCZ}}_{a,b,N} \prod_{k=1}^{\frac{N-1}{2}} \overline{\text{CCZ}}_{2k-1,2k,N}
\end{equation}
To implement $\overline{\text{CCZ}}_{2k-1,2k,N}$, we must thus isolate out the logical qubits we wish to act on. This can be done by using the gadget $I_{2K-1,a}$ to move logical information on qubit $K$ temporarily onto the ancilla qubit $a$. Note now that qubits $2K-1$ and $b$ now do not carry logical information. Thus, we can prepare both these qubits in $|\bar{0}\rangle$. This then means that $\overline{\text{CCZ}}_{a,b,N} \overline{\text{CCZ}}_{2K-1,2K,N}$ acts trivially, and so applying $CCZ$ transversally implements the logical operator in Eq.~\ref{A2}.
\begin{equation}\label{A2}
\prod_{k=1}^{K-1}\overline{\text{CCZ}}_{2k-1,2k,N} \prod_{k=K+1}^{\frac{N-1}{2}} \overline{\text{CCZ}}_{2k-1,2k,N}
\end{equation}
We can then return the logical information being stored in qubit $a$ back to qubit $2K-1$ by applying the gadget $I_{a,2K-1}$. With $b$ still in state $|\bar{0}\rangle$ we can now implement $CCZ$ transversally again. This then implements the operator in Eq.~\ref{A3}.
\begin{equation}\label{A3}
\prod_{k=1}^{\frac{N-1}{2}} \overline{\text{CCZ}}_{2k-1,2k,N}
\end{equation} 
Thus, the net effect of applying all of these operations is that the required operator $\overline{CCZ}_{2K-1,2K,N}$ has been applied, since this is equivalent to the operator in Eq.~\ref{A4}.
\begin{equation}\label{A4}
\prod_{k=1}^{K-1} \overline{\text{CCZ}}_{2k-1,2k,N} \prod_{k=K+1}^{\frac{N-1}{2}} \overline{\text{CCZ}}_{2k-1,2k,N} \prod_{k=1}^{\frac{N-1}{2}}\overline{CCZ}_{2k-1,2k,N}
\end{equation}
The full circuit used to perform this operator, in the case $N=5$ is shown in Fig.~\ref{CCZ}.

\begin{figure}
\includegraphics[width=0.48\textwidth]{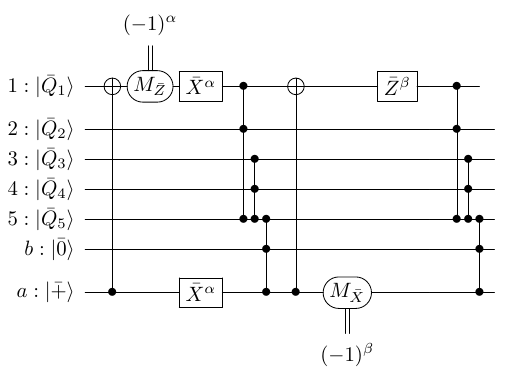}
\caption{Circuit that implements $\overline{\text{CCZ}}_{125}$ in the case of the scheme with five logical qubits.\label{CCZ}}
\end{figure}

We now consider how to implement logical $\bar{H}$ operators on each logical qubit, $x$. To do this, we begin by noting that we may still apply $\overline{\text{CZ}}_{N,x}$ (and, hence, the gadget $H_{N,x}$) for any logical qubit or ancilla $x \neq N$. We now must just check that we can implement $\bar{H}_{ab}$ and $\bar{H}_{ba}$ to be able to reconstruct the circuits used to implement the logical Hadamard in the three logical qubit case. We note that this can be done relatively simply. In particular, note that applying transversal $\text{CZ}_{12}$ across a membrane contained by the threaded defect now implements the operator in Eq.~\ref{A5}.
\begin{equation}\label{A5}
\overline{\text{CZ}}_{ab} \prod_{k=1}^{\frac{N-1}{2}} \overline{\text{CZ}}_{2k-1,2k}
\end{equation}
By first preparing the target ancilla qubit (i.e.~$b$ for $H_{ab}$ or $a$ for $H_{ba}$) in $|\bar{0}\rangle$, we can ensure that this becomes $\prod_{k=1}^{\frac{N-1}{2}} \text{CZ}_{2k-1,2k}$. Preparing that qubit now in $|\bar{+}\rangle$ and applying $CZ_{12}$ transversally again implements the operator in Eq.~\ref{A6}.
\begin{equation}\label{A6}
\overline{CZ}_{ab} \prod_{k=1}^{\frac{N-1}{2}} \text{CZ}_{2k-1,2k}
\end{equation}
The net effect is then $\overline{\text{CZ}}_{ab}$. This allows us to implement $H_{ab}$ or $H_{ba}$ and so to realise the circuit for implementing logical Hadamards described in the previous subsection.

To conclude, we now note that we may swap any logical qubit, $x$, with qubit $N$ by the operator in Eq.~\ref{A7}.
\begin{equation}\label{A7}
\text{SWAP}_{x,N}=\bar{H}_N \overline{\text{CZ}}_{x,N}\bar{H}_N\bar{H}_x \overline{\text{CZ}}_{x,N}\bar{H}_x\bar{H}_N \overline{\text{CZ}}_{x,N}\bar{H}_N
\end{equation}
This means that we may swap any two logical qubits $x$ and $y$ by implementing $\text{SWAP}_{x,N} \text{SWAP}_{y,N}  \text{SWAP}_{x,N}$. Thus, we may implement $\overline{\text{CCZ}}_{xyz}$ between any three logical qubits $x,y,z$.


\begin{thebibliography}{99}
\bibitem{Eastin}
B.~Eastin and E.~Knill, ``Restrictions on Transversal Encoded Gate Sets'', Phys.~Rev.~Lett., \textbf{102}, 110502, 2009.
\bibitem{Bravyi}
S.~Bravyi and R.~K{\"o}nig, ``Classification of topologically protected gates for local stabilizer codes'',  Phys.~Rev.~Lett., \textbf{110}, 170503, 2013.
\bibitem{Pastawski}
F.~Pastawski and B.~Yoshida, ``Fault-tolerant logical gates in quantum error-correcting codes'', Phys.~Rev.~A, \textbf{91}, 012305, 2015.
\bibitem{Webster}
P.~Webster and S.~D.~Bartlett, ``Locality-preserving logical operators in topological stabilizer codes'', Phys.~Rev.~A, \textbf{97}, 012330, 2018.
\bibitem{Vasmer}
M.~Vasmer and D.~E.~Browne, ``Three-dimensional surface codes: Transversal gates and fault-tolerant architectures'', Phys.~Rev.~A, \textbf{100}, 012312, 2019.
\bibitem{BombinGCC}
H.~Bombin, ``Dimensional jump in quantum error correction'', New J.~Phys., \textbf{18}, 043038, 2016.
\bibitem{Bombin2018}
H.~Bombin, ``2D quantum computation with 3D topological codes'', arXiv:1810.09571, 2018.
\bibitem{BrownJIT}
B.~J.~Brown, ``A fault-tolerant non-Clifford gate for the surface code in two dimensions'', Sci.~Adv., \textbf{6}, eaay4929, 2020.
\bibitem{BombinNoBraid}
H.~Bombin and M.~A.~Martin-Delgado, ``Topological Computation without Braiding'', Phys.~Rev.~Lett., \textbf{98}, 160502, 2007.
\bibitem{Kubica}
A.~Kubica and M.~E.~Beverland, ``Universal transversal gates with color codes: A simplified approach'', Phys.~Rev.~A, \textbf{91}, 032330, 2015.
\bibitem{Jochym}
T.~Jochym-O'Connor and S.~D.~Bartlett, ``Stacked codes: Universal fault-tolerant quantum computation in a two-dimensional layout'', Phys.~Rev.~A, \textbf{93}, 022323, 2016.
\bibitem{Roberts}
B.~J.~Brown and S.~Roberts, ``Universal fault-tolerant measurement-based quantum computation'', arXiv:1811.11780, 2018.
\bibitem{Campbell}
E.~T.~Campbell, B.~M.~Terhal and C.~Vuillot, ``Roads towards fault-tolerant universal quantum computation'', Nature, \textbf{549}, 172--179, 2017.
\bibitem{MSD}
S.~Bravyi and A.~Kitaev, ``Universal quantum  computation with ideal Clifford gates and noisy ancillas'', Phys.~Rev.~A, \textbf{71}, 022316, 2005.
\bibitem{Fowler2}
A.~G.~Fowler, M.~Mariantoni, J.~M.~Martinis and A.~N.~Cleland, ``Surface codes: Towards practical large-scale quantum computation'', Phys.~Rev.~A, \textbf{86}, 032324, 2012.
\bibitem{Raussendorf1}
R.~Raussendorf and J.~Harrington, ``Fault-tolerant quantum computation with high threshold in two dimensions'', Phys.~Rev.~Lett., \textbf{98}, 190504, 2007.
\bibitem{Raussendorf2}
R.~Raussendorf, J.~Harrington and K.~Goyal, ``Topological fault-tolerance in cluster state quantum computation'', New J.~Phys., \textbf{9}, 199, 2007.
\bibitem{BombinHole}
H.~Bombin and M.~A.~Martin-Delgado, ``An interferometry-free protocol for demonstrating topological order'', Phys.~Rev.~B, \textbf{78}, 165128, 2008.
\bibitem{Fowler1}
A.~G.~Fowler, A.~M.~Stephens and P.~Groszkowski, ``High threshold universal quantum computation on the surface code'', Phys.~Rev.~A, \textbf{80}, 052312, 2009.
\bibitem{Fowler3}
A.~G.~Fowler, ``2-D color quantum computation'', Phys.~Rev.~A, \textbf{83}, 042310, 2011.
\bibitem{Hastings}
M.~B.~Hastings and A.~Geller, ``Reduced space-time and time costs using dislocation codes and arbitrary ancillas'', Quantum Information and Computation, \textbf{15}, 0962, 2015.
\bibitem{Brell}
C.~G.~Brell, ``Generalized cluster states based on finite groups'', New J.~Phys., \textbf{17}, 023029, 2015.
\bibitem{Brown}
B.~J.~Brown, K.~Laubscher, M.~S.~Kesselring and J.~R.~Wootton, ``Poking holes and cutting corners to achieve Clifford gates with the surface code'', Phys.~Rev.~X, \textbf{7}, 021029, 2017.
\bibitem{Bombin}
H.~Bombin, ``Topological Order with a Twist: Ising Anyons from an Abelian Model'', Phys.~Rev.~Lett., \textbf{105}, 030403, 2010.
\bibitem{Kesselring}
M.~S.~Kesselring, F.~Pastawski, J.~Eisert and B.~J.~Brown, ``The boundaries and twist defects of the color code and their applications to topological quantum computation'', Quantum, \textbf{2}, 101, 2018.
\bibitem{Scruby}
T.~R.~Scruby and D.~E.~Browne, ``A Hierarchy of Anyon Models Realised by Twists in Stacked Surface Codes'', Quantum, \textbf{4}, 251, 2020.
\bibitem{Bombin2}
H.~Bombin, ``Clifford gates by code deformation'', New J.~Phys., \textbf{13}, 043005, 2011.
\bibitem{Cong}
I.~Cong, M.~Cheng and Z.~Wang, ``Universal quantum computation with gapped boundaries'', Phys.~Rev.~Lett., \textbf{119}, 170504, 2017.
\bibitem{Yoder}
T.~J.~Yoder and I.~H.~Kim, ``The surface code with a twist'', Quantum, \textbf{1}, 2, 2017. 
\bibitem{BarkeshliGenon1}
M.~Barkeshli, C.-M.~Jian and X.-L.~Qi, ``Genons, twist defects, and projective non-Abelian braiding statistics'', Phys.~Rev.~B, \textbf{87}, 045130, 2013.
\bibitem{BarkeshliGenon2}
M.~Barkeshli, C.-M.~Jian and X.-L.~Qi, ``Theory of defects in Abelian topological states'', Phys.~Rev.~B, \textbf{88}, 235103, 2013.
\bibitem{BarkeshliTheory1}
M.~Barkeshli, C.-M.~Jian and X.-L.~Qi, ``Classification of topological defects in abelian topological states'', Phys.~Rev.~B, \textbf{88}, 241103(R), 2013.
\bibitem{BarkeshliTheory2}
M.~Barkeshli, P.~Bonderson, M.~Cheng and Z.~Wang, ``Symmetry Fractionalization, Defects, and Gauging of topological phases'', arXiv:1410.4540, 2014.
\bibitem{Freedman}
M.~H.~Freedman, M.~Larsen, Z.~Wang, “A modular functor which is universal for quantum
computation”, Commun. Math. Phys., \textbf{227(3)}, 605–622, 2002.
\bibitem{Kitaev}
A.~Kitaev, ``Anyons in an exactly solved model and beyond'', Annals of Physics, \textbf{321}, 2–111, 2006.
\bibitem{Levin}
M.~Levin and X.-G.~Wen, ``Fermions, strings and gauge fields in lattice spin models'', Phys.~Rev.~B,  \textbf{67}, 245316, 2003.
\bibitem{Vijay}
S.~Vijay, J.~Haah and L.~Fu, ``Fracton topological order, generalized lattice gauge theory, and duality'', Phys.~Rev.~B, \textbf{94}, 235157, 2016.
\bibitem{Shirley1}
W.~Shirley, K.~Slagle and X.~Chen, ``Foliated fracton order in the checkerboard model'', Phys.~Rev.~B, \textbf{99}, 115123, 2019.
\bibitem{Dua}
A.~Dua, I.~H.~Kim, M.~Cheng and D.~J.~Williamson, ``Sorting topological stabilizer models in three dimensions'', Phys.~Rev.~B, \textbf{100}, 155137, 2019.
\bibitem{BombinSingleShot}
H.~Bombin, ``Single-shot fault-tolerant quantum error correction'', Phys.~Rev.~X, \textbf{5}, 031043, 2015.
\bibitem{Aharanov}
D.~Aharanov and M.~Ben-Or, ``Fault-tolerant quantum computation with constant error'', STOC '97 Proceedings of the twenty-ninth annual ACM symposium on Theory of computing, pp.~176--188, 1997.
\bibitem{Terhal}
B.~M.~Terhal, ``Quantum Error Correction for Quantum Memories'', Rev.~Mod.~Phys., \textbf{87}, 307, 2015.
\bibitem{Gottesman}
D.~Gottesman, ``Stabilizer Codes and Quantum Error Correction'', PhD Thesis, California Institute of Technology, 1997.
\bibitem{Kitaev2}
A.~Yu.~Kitaev, ``Fault-tolerant quantum computation by anyons'', Ann.~Phys., \textbf{303}, 2, 2003.
\bibitem{Yoshida}
B.~Yoshida, ``Topological color code and symmetry-protected topological phases'', Phys.~Rev.~B, \textbf{91}, 245131, 2015.
\bibitem{Chamon}
C.~Chamon, ``Quantum Glassiness in Strongly Correlated Clean Systems: An Example of Topological Overprotection'', Phys.~Rev.~Lett., \textbf{94}, 040402, 2005.
\bibitem{HaahCubic}
J.~Haah, ``Local stabilizer code in three dimensions without string logical operators'', Phys.~Rev.~A, \textbf{83}, 042330, 2011.
\bibitem{Shirley2}
W.~Shirley, K.~Slagle and X.~Chen, ``Fractional excitations in foliated fracton phases'',  Annals of Physics, \textbf{410}, 167922, 2019.
\bibitem{Else}
D.~V.~Else and C.~Nayak, ``Cheshire charge in (3+1)-dimensional topological phases'', Phys.~Rev.~B, \textbf{96}, 045136, 2017.
\bibitem{BombinS}
H.~Bombin, ``Transversal gates and error propagation in 3D topological codes'', arXiv:1810:09575, 2018.
\bibitem{YoshidaB}
B.~Yoshida, ``Topological phases with generalized global symmetries'', Phys.~Rev.~B, \textbf{93}, 155131, 2016.
\bibitem{YoshidaC}
B.~Yoshida, ``Gapped boundaries, group cohomology and fault-tolerant logical gates'', Annals of Physics, \textbf{337}, 387, 2017.
\bibitem{Kubica2018}
A.~Kubica and B.~Yoshida, ``Ungauging quantum error-correcting codes'', arXiv:1805.01836, 2018.
\bibitem{BravyiKitaev}
S.~Bravyi and A.~Yu.~Kitaev, ``Quantum codes on a lattice with boundary'', arXiv:quant-ph/9811052v1, 1998.
\bibitem{ColourCode}
H.~Bombin and M.~A.~Martin-Delgado, ``Topological Quantum Distillation'', Phys.~Rev.~Lett., \textbf{97}, 180501, 2006.
\bibitem{BombinMD}
H.~Bombin and M.~A.~Martin-Delgado, ``Quantum measurements and gates by code deformation'', J.~Phys.~Math.~Theor., \textbf{42}, 095302, 2009.
\bibitem{Zhu}
G.~Zhu, A.~Lavasani and M.~Barkeshli, ``Fast universal logical gates on topologically encoded qubits at arbitrarily large code distances'',  arXiv:1806.02358v2, 2018.
\bibitem{Zhu2}
A.~Lavasani, G.~Zhu and M.~Barkeshli, ``Universal logical gates with constant overhead: instantaneous Dehn twists for hyperbolic quantum codes'', Quantum, \textbf{3}, 180, 2019.
\bibitem{Horsman}
C.~Horsman, A.~G.~Fowler, S.~Devitt and R.~Van Meter, ``Surface code quantum computing by lattice surgery'', New J.~Phys., \textbf{14}, 123011, 2012.
\bibitem{Wang}
C.~Wang and M.~Levin, ``Braiding statistics of loop excitations in three dimensions'', Phys.~Rev.~Lett., \textbf{113}, 080403, 2014.
\bibitem{Escobar}
N.~Escobar-Vel{\' a}squez, C.~Galindo and Z.~Wang, ``Braid group representations from braiding gapped boundaries of Dijkgraaf-Witten theories'', Pacific Journal of Mathematics, \textbf{300}, 1--16, 2019.
\bibitem{Bombin4}
H.~Bombin, ``An introduction to Topological Quantum Codes'' in ``Quantum Error Correction'', eds: D.~A.~Lidar and T.~A.~Brun, Cambridge University Press, New York, 2013.
\bibitem{Mesaros}
A.~Mesaros, Y.~B.~Kim and Y.~Ran, ``Changing topology by topological defects in three-dimensional topologically ordered phases'', Phys.~Rev.~B, \textbf{88}, 035141, 2013.
\bibitem{Beverland}
M.~E.~Beverland, O.~Buerschaper, R.~Koenig, F.~Pastawski, J.~Preskill and S.~Sijher, ``Protected gates for topological quantum field theories'', J.~Math.~Phys., \textbf{57}, 022201, 2016.
\bibitem{Haah}
J.~Haah, ``Commuting Pauli Hamiltonians as maps between free modules'',  Comm.~Math.~Phys., \textbf{324}, 351--399, 2013.
\bibitem{Roberts4DSC}
S.~Roberts, B.~Yoshida, A.~Kubica and S.~D.~Bartlett, ``Symmetry protected topological order at nonzero temperature'', Phys.~Rev.~A, \textbf{96}, 022306, 2017.
\bibitem{YoderBacon}
T.~J.~Yoder, ``Universal fault-tolerant quantum computation with Bacon-Shor codes'', arXiv:1705.01686, 2017.
\bibitem{KubicaUnfold}
A.~Kubica, B.~Yoshida and F.~Pastawski, ``Unfolding the color code'', New J.~Phys., \textbf{17}, 083026, 2015.
\bibitem{Litinski}
D.~Litinski, ``Magic State Distillation: Not as Costly as You Think'', Quantum, \textbf{3}, 205, 2019.
\end{thebibliography}
\end{document}